\documentclass[draft,12pt,reqno,a4paper]{amsart}

\usepackage[margin=3cm,footskip=1cm]{geometry}

\usepackage{amssymb} \usepackage{amsmath} \usepackage{amsfonts}
\usepackage{amsthm} \usepackage{mathtools}

\usepackage{enumerate} \usepackage[latin1]{inputenc}
\usepackage[shortlabels]{enumitem}

\usepackage{graphicx} \usepackage{verbatim}

\DeclareMathOperator*{\slim}{s--lim}

\newcommand{\supp}{\operatorname{supp}}

\newcommand{\dist}{{\operatorname{dist}}}

\newcommand{\Ran}{{\operatorname{Ran}}}

\newcommand{\ad}{{\operatorname{ad}}}
\newcommand{\N}{{\mathbb{N}}} 
\newcommand{\R}{{\mathbb{R}}} \newcommand{\Z}{{\mathbb{Z}}}
\newcommand{\C}{{\mathbb{C}}}

 \renewcommand{\c}{{\rm c}}
\newcommand{\e}{{\rm e}} \newcommand{\ess}{{\rm ess}}
 \renewcommand{\i}{{\rm i}}
\renewcommand{\d}{{\rm d}}

\newcommand{\pupo}{{\rm pp}}

\renewcommand{\Re}{{\rm Re}\,} \renewcommand{\Im}{{\rm Im}\,}

\DeclarePairedDelimiter\inp\langle\rangle

\newcommand\parb[2][]{#1 \big ( #2#1\big )} \newcommand\parbb[2][]{#1
  \Big ( #2#1\Big )}

 \renewcommand{\exp}{{\rm exp}}
\newcommand{\mand}{\text{ and }} \newcommand{\mfor}{\text{ for }}
\newcommand{\mforall}{\text{ for all }}
  
\newcommand{\mif}{\text{ if }} \newcommand{\mon}{\text{ on }}

\newcommand{\bY}{{\bf Y}} \newcommand{\bX}{{\bf X}}

\newcommand{\vA}{{\mathcal A}} \newcommand{\vB}{{\mathcal B}}
\newcommand{\vC}{{\mathcal C}} \newcommand{\vD}{{\mathcal D}}

 \newcommand{\vH}{{\mathcal H}}

\newcommand{\vT}{{\mathcal T}} 
\newcommand{\vV}{{\mathcal V}}

\newcommand\myboxed[1]{{\fboxsep=0.4pt\boxed{#1}}}

\theoremstyle{plain}
\newtheorem{thm}{Theorem}[section]
\newtheorem{proposition}[thm]{Proposition}
\newtheorem{lemma}[thm]{Lemma} \newtheorem{corollary}[thm]{Corollary}
\theoremstyle{definition}

 \newtheorem{cond}[thm]{Condition}

\newtheorem{remarks}[thm]{Remarks}
 \newtheorem*{remarks*}{Remarks}
\newtheorem*{remark*}{Remark}

\numberwithin{equation}{section}

\title[Absence of positive eigenvalues for hard-core $N$-body
systems]{Absence of positive eigenvalues for hard-core $N$-body
  systems}

\author{K. Ito}
\address[K. Ito]{Graduate School of Pure and Applied Sciences, University of Tsukuba\\
  1-1-1 Tennodai, Tsukuba Ibaraki, 305-8571 Japan}
\email{ito-ken@math.tsukuba.ac.jp}

\author{E. Skibsted} \address[E. Skibsted]{Institut for Matematiske
  Fag \\
  Aarhus Universitet\\ Ny Munkegade 8000 Aarhus C, Denmark}
\email{skibsted@imf.au.dk}

\begin{document}
\begin{abstract} We show absence of positive eigenvalues for
  generalized $2$-body hard-core Schr\"{o}dinger operators under the
  condition of bounded strictly convex obstacles. A scheme for showing
  absence of positive eigenvalues for generalized $N$-body hard-core
  Schr\"{o}dinger operators, $N\geq 2$, is presented. This scheme
  involves high energy resolvent estimates, and for $N=2$ it is
  implemented by a Mourre commutator type method. A particular example
  is the Helium atom with the assumption of infinite mass and finite
  extent nucleus.
\end{abstract}

\maketitle
\tableofcontents

\section{Introduction and results}\label{sec:introduction}
Consider the $N$-body Schr\"{o}dinger operator 
\begin{align}\label{eq:Hami_1}
  H=\sum_{j = 1}^N \Big (-\frac{1}{2m_j}\Delta_{x_j}+V^{\rm
    ncl}_{j}(x_j)\Big ) + \sum_{1 \le i<j \le N} V_{ij}^{\rm elec}(x_i
  -x_j)
\end{align} for a system of $N$ $d$-dimensional particles in the
exterior of a bounded strictly convex obstacle $\Theta_1\subset
\mathbb{R}^d$ (for $N=1$ the last term is omitted). Whence $H$ is an
operator on the Hilbert space $L^2(\Omega)$; $\Omega=(\Omega_1)^N$,
$\Omega_1=\R^d\setminus \overline \Theta_1$. It is defined more
precisely by imposing the Dirichlet boundary condition. This operator
models a system of $N$ $d$-dimensional charged particles interacting
with a fixed nucleus of finite extent, for example a ball (or possibly
a somewhat  deformed ball). In particular (assuming $0\in \Theta_1$) we
could have Coulomb interactions $V^{\rm ncl}_{j}(y)=q_jq^{\rm
  ncl}|y|^{-1}$ and $V_{ij}^{\rm elec}(y)=q_iq_j|y|^{-1}$ in dimension
$d\geq 2$. We address the problem of proving absence of positive
eigenvalues. While this property is well-known for the one-body
problem it is open for $N\geq 2$. We introduce for obstacle problems
of this type a general procedure involving high energy resolvent
estimates for effective sub-Hamiltonians. We show that this scheme can
be implemented for the case $N=2$. In this case essentially such an 
effective sub-Hamiltonian is a one-body Hamiltonian for an exterior
region. The result is shown in the so-called generalized $2$-body
hard-core framework.

\subsection{Usual generalized $N$-body
  systems}\label{subsec:Usual generalized $N$-body
  systems}  We will work in a generalized framework. We first review the analogue of this
without obstacles, i.e. with ``soft potentials''. This is given by
real finite dimensional vector space $\bX$ with an inner product $q$,
i.e.  $(\bX,q)$ is Euclidean
space,  and a
finite family of subspaces $\{\bX_a| \ a\in \vA\}$ closed with respect
to intersection. We refer to the elements of $\vA$ as {\it cluster
decompositions} (not to be motivated here). The orthogonal complement of $\bX_a$ in $\bX$ is
denoted $\bX^a$, and correspondingly we decompose $x=x^a\oplus x_a\in
\bX^a\oplus \bX_a$. We order $\vA$ by writing $a_1\subset a_2$ if
$\bX^{a_1}\subset \bX^{a_2}$. It is assumed that there exist
$a_{\min},a_{\max}\in \vA$ such that $\bX^{a_{\min}}=\{0\}$ and
$\bX^{a_{\max}}=\bX$. Let $\vB=\vA\setminus\{a_{\min}\}$. The length
of a chain of cluster decompositions $a_1\subsetneq \cdots \subsetneq
a_k$ is the number $k$. Such a  chain is said to connect $a=a_1$ and
$b=a_k$. The maximal length of all chains connecting a given $a\in
\vA\setminus\{a_{\max}\}$ and $a_{\max}$ is denoted by $\# a$. We
define $\# a_{\max}=1$ and denoting $\# a_{\min}=N+1$ we say the
family $\{\bX^a|a\in\vA\}$ is of $N$-body type. Whence the generalized
$2$-body framework is characterized by the condition $\bX_a\cap
\bX_b=0$ for $a,b\neq a_{\min}$, $a\neq b$.

The $N$-body Schr\"odinger operator $H$ introduced above (now
considered without an obstacle) can be written on the form
\[
H = H_0 + V
\]
where $2H_0$ is (minus) the Laplace-Beltrami operator on the 
space
\begin{align*}
 \bX=(\R^d)^N, \quad q=\sum_{j=1}^N m_j|x_j|^2, 
\end{align*}
 $V=V(x) = \sum_{b\in\vB} V_{b}(x^{b})$ and indeed
the relevant family $\{\bX^a|a\in\vA\}$ of subspaces as discussed
above is of $N$-body type.
However this is just one example of a generalized $N$-body
Schr\"odinger operator. The  general construction of such an operator
$H$ is
similar, and under the following condition it is well-defined with
form domain given by the Sobolev space $H^1(\bX)$, cf.  \cite[Theorem X.17]{RS}.

\begin{cond}\label{cond:smooth1} 
  There exists $\varepsilon >0$ such that for potential $V_b$,
  $b\in\vB$, there is a splitting $V_b=V_b^{(1)}+V_b^{(2)}$, where
  \begin{enumerate}[label=(\arabic*)]
  \item \label{item:cond1} $V_b^{(1)}$ is smooth and
    \begin{equation}
      \label{eq:1k}
      \partial ^\alpha_yV_b^{(1)}(y)=O\big(|y|^{-\varepsilon-|\alpha|}\big ).
    \end{equation}
  \item \label{item:cond12} $V_b^{(2)}$ is compactly supported and
    \begin{equation}
      \label{eq:2k}
      (-\Delta+1)^{-1/2}V_{b}^{(2)}(-\Delta+1)^{-1/2}\text{ is compact on
      }L^2(\R^{\dim \bX ^b}_y).
    \end{equation}
  \end{enumerate}
\end{cond}

Let $ -\Delta^{a}=(p^a)^2$ and $-\Delta_{a}=p_a^2$ denote (minus) the
Laplacians on $L^2(\bX^a)$ and $L^2(\bX_a)$, respectively. Here
$p^a=\pi^ap$ and $p_a=\pi_ap$ denote the internal (i.e. within
clusters) and the inter-cluster components of the momentum operator
$p=-\i\nabla$, respectively.  For $a\in \vB$, denote
\begin{align*}
  V^a(x^a)&=\sum_{b\subset a}   V_{b}(x^{b}),\\
  H^{a} &= -\tfrac12\Delta^{a} + V^a(x^a), \\
  H_{a}&= H^{a}- \tfrac12 \Delta_{a} ,\\
  I_{a}(x) &= \sum_{b\not\subset a} V_{b}(x^{b}).
\end{align*}
We define $H^{a_{\min}}=0$ on $L^2(\bX^{a_{\min}}):=\C$. The operator
$H^a$ is the sub-Hamiltonian associated with the cluster decomposition
$a$ and $I_a$ is the sum of all inter-cluster interactions. The
detailed expression of $H^a$ depends on the choice of coordinates on
$\bX^a$.

In a natural way we have sub-Hamiltonians $H^a$ and ``inter-cluster''
Hamiltonians $H_a=H^a\otimes I+ I \otimes \tfrac12 p^2_a$. Given a
family $\{\bX^a|a\in\vA\}$ of $N$-body type and imposing Condition
\ref{cond:smooth1} the generalized $N$-body Hamiltonian is
$H=H^{a_{\max}}$.

Let
\[
\vT = \cup_{a\in\vA, \# a\ge 2} \;\sigma_{\pupo}( H^a)
\]
be the set of thresholds of $H$.  The HVZ theorem \cite[Theorem
XIII.17]{RS} gives the bottom of the essential spectrum $\Sigma_2
:=\inf \sigma_{\ess}(H) $ of $H$ by the formula
\begin{equation}
  \label{eq:3k}
  \Sigma_2 = \min_{a\in\vA\setminus\{a_{\max}\}} \inf\sigma( H^a) =
  \min_{a\in\vA, \# a = 2} \inf\sigma( H^a).
\end{equation}

It is also well-known that under rather general conditions $H$ does
not have positive eigenvalues and the negative eigenvalues can at most
accumulate at the thresholds from below, see \cite{FH} and \cite{Pe}.

\subsubsection{Graf vector field}
We give a brief review of the construction of a family of conjugate
operators for $N$-body Hamiltoninans originating from \cite{Sk1}. A
slightly different proof appears in \cite{Sk3}. This construction is
based on the vector field invented by Graf \cite{Gra} which is a
vector field satisfying the following properties, cf. \cite[Lemma
4.3]{Sk3}.  We use throughout the paper the notation $\langle x
\rangle = \sqrt{ x^2+1}$ and $\N_0=\N\cup\{0\}$.
\begin{lemma}
  \label{lemma:vector field} There exist on $\bX$ a smooth vector
  field $\omega$ with symmetric derivative $\omega_*$ and a partition
  of unity $\{\tilde q_a\}$ indexed by $a\in \vA$ and consisting of
  smooth functions, $0\leq \tilde q_a\leq 1$, such that for some
  positive constants $r_1$ and $r_2$
  \begin{enumerate}[\normalfont (1)]
  \item \label{item:4k}$\omega_*(x)\geq\sum _a\pi_a \tilde q_a.$
  \item \label{item:5k}$\omega^a(x)=0\text{ if }|x^a|<r_1.$
  \item\label{item:6k} $|x^b|>r_1\mon \supp(\tilde q_a)\mif b\not
    \subset a.$
  \item\label{item:7k} $|x^a|<r_2\mon \supp(\tilde q_a).$
  \item\label{item:7ik} For all $\alpha\in \N_0^{\dim \bX}$ and $k\in
    \N_0$ there exist $C\in \R$:
    \begin{equation}\label{eq:boun_constr}
      |\partial ^\alpha_x\tilde q_a|+|\partial ^\alpha_x(x\cdot
      \nabla)^k\parb{\omega(x)-x}|\leq C.
    \end{equation}
  \end{enumerate}
\end{lemma}

  Now, proceeding as in \cite{Sk3}, we introduce the rescaled vector
  field $\omega_R(x):=R\omega(\frac xR)$ and the corresponding
  operator
  \begin{equation*}
    A=A_R=\omega_R(x)\cdot p+p\cdot \omega_R(x);\;R>1.
  \end{equation*} We also   introduce the  function $d:\R\to\R$  by
  \begin{equation}\label{eq:44k}
    d(E)=\begin{cases}\inf _{\tau\in \vT (E)}(E-\tau),\;\vT
      (E):=\vT\cap \,]-\infty,E]\neq \emptyset,\\
      1,\;\vT
      (E)=\emptyset.
    \end{cases}
  \end{equation} These devices enter into the following Mourre
  estimate.  We remark that  all  inputs needed for the proof are
  stated in Lemma \ref{lemma:vector field}  and that although
  \cite[Corollary 4.5]{Sk3} is stated for relatively operator compact
  potentials the proof of \cite{Sk3} generalizes to include the class of relatively form compact
  potentials of Condition \ref{cond:smooth1}. For a different proof we
  refer to \cite{Gri}.
  \begin{lemma}
    \label{lemma:Mourre1} For all $E\in \R$ and $\epsilon>0$ the
    exists $R_0>1$ such that for all $R\geq R_0$ there is a
    neighbourhood $\vV$ of $E$ and a compact operator $K$ on
    $L^2(\bX)$ such that
    \begin{equation}\label{eq:40k}
      f(H)^*\i [H,A_R]f(H)\geq f(H)^*\{4d(E)-\epsilon-K\}f(H)\mforall
      f\in C^\infty_{\c}(\vV).
    \end{equation}
  \end{lemma}
  Here the commutator is given by \eqref{eq:mourre comm} stated
  below. The possibly existing local singularities of the potential
  does not enter (for $R$ large) due to Lemma \ref{lemma:vector field}
  \ref{item:5k}. This feature motivates application to hard-core
  models, see Subsection \ref{subsec:Generalized hard-core $N$-body
    systems}.

  Two of the consequences of a Mourre estimate like the one stated
  above are that the set of thresholds $\vT$ is closed and countable
  and that the eigenvalues of $H$ can at most accumulate at $\vT$. We
  discuss a third consequence, decay of non-threshold eigenstates, in
  Subsection \ref{subsec:Generalized hard-core $N$-body systems}.
  
\subsection{Generalized $N$-body hard-core  
  systems}\label{subsec:Generalized hard-core  $N$-body
  systems}  The generalized hard-core model is a modification for the
above model. For the generalized hard-core  model we are   given  for
each $a\in \vB$ an open subset
$\Omega_a\subset \bX^a$ with $\bX^a\setminus \Omega_a$ compact,
possibly $\Omega_a= \bX^a$. Let for $a_{\min}\neq b\subset a$
\begin{equation*}
  \Omega^a_b=\parb{ \Omega_b+\bX_b}\cap \bX^a=\Omega_b+\bX_b\cap \bX^a,
\end{equation*} and for $a \neq a_{\min}$ 
\begin{equation*}
  \Omega^a=\cap_{a_{\min}\neq b\subset a} \Omega^a_b.
\end{equation*} We define $\Omega^{a_{\min}}=\{0\}$.

\begin{cond}\label{cond:smooth2} 
  There exists $\varepsilon >0$ such that for all $b\in\vB$ there is a
  splitting $V_b=V_b^{(1)}+V_b^{(2)}$, where
  \begin{enumerate}[label=(\arabic*)]
  \item \label{item:cond12b} $V_b^{(1)}$ is smooth on the closure of
    $\Omega_b$ and
    \begin{equation}
      \label{eq:1k2}
      \partial ^\alpha_yV_b^{(1)}(y)=O\big(|y|^{-\varepsilon-|\alpha|}\big ).
    \end{equation}
  \item \label{item:cond122} $V_b^{(2)}$ vanishes outside a bounded
    set in $\Omega ^b$ and
    \begin{equation}
      \label{eq:2k2}
      V_{b}^{(2)}\in \vC\big (H^1_0(\Omega_b),H^1_0(\Omega_b)^*\big ).
    \end{equation}
  \end{enumerate}
\end{cond}
Here and henceforth, given Banach spaces $X_1$ and $X_2$, the notation
$\vC\big (X_1, X_2\big )$ and $\vB\big (X_1, X_2\big )$ refers to the
set of compact and the set of bounded operators $T:X_1\to X_2$,
respectively.

We consider for $a\in\vB$ the Hamiltonian $H^a=-\tfrac 12
\Delta_{x^a} +V^a$ on $L^2(\Omega^a)$ with Dirichlet boundary
condition on $\partial \Omega^a$, in particular $H=\tfrac 12 p^2+V$
on $L^2(\Omega)$ with Dirichlet boundary condition on $\partial
\Omega$ where $ \Omega:=\Omega^{a_{\max}}$.  The corresponding form
domain is the Sobolev space $H^1_0(\Omega^a)$. Due to the continuous
embedding $H_0^1(\Omega^a)\subset H_0^1(\Omega^a_b)$ for
$a_{\min}\neq b\subset a$ we conclude that indeed $H^a$ is
self-adjoint, cf. \cite[Theorem X.17]{RS}. Again we define
$H^{a_{\min}}=0$, and the set of thresholds is also given as in
Subsection \ref{subsec:Usual generalized $N$-body systems}. We note
that one can replace the Hilbert space $L^2(\bX)$ in Lemma
\ref{lemma:Mourre1} by $\vH:=L^2(\Omega)$ and then obtain a Mourre
estimate for the present Hamiltonian $H$, cf. \cite[Theorem
2.4]{Gri}. All what is needed for this is to make sure that $R>1$ is
so large that the rescaled Graf vector field $\omega_R$ either
vanishes or acts tangentially on the boundary $\partial \Omega$. The
latter is doable due to Lemma \ref{lemma:vector field} \ref{item:5k}.

  According to \cite[Theorem 2.5(1)]{Gri} non-threshold eigenstates
  decay exponentially at rates determined by thresholds above the
  corresponding eigenvalues. This is a consequence of the hard-core
  Mourre estimate by arguments similar to the ones of \cite{FH} for
  usual $N$-body Hamiltonians. In \cite{Gri} Griesemer states as an
  open problem absence of positive eigenvalues under an additional
  connectedness condition. This is the problem we shall address in the
  present paper. The pattern of proof of \cite{FH} does not work
  except the following induction scheme: For $N=1$ absence of positive
  eigenvalues follows from various papers (assuming that
  $\Omega\subset \bX$ is connected), for example most recently
  \cite{IS}. For $N\geq 2$ we could suppose by induction that the
  result holds for sub-Hamiltonians, whence that there are no positive
  thresholds. Using the hard-core Mourre estimate in a similar way as
  for soft potentials \cite{FH, Gri, IS} we then deduce that an
  eigenstate with corresponding positive eigenvalue would decay
  super-exponentially, cf. \cite[Theorem 2.5(1)]{Gri}. This would be
  derived in terms of the potential function $r$ discussed
  below. Whence for any such eigenstate $\phi$ (i.e. corresponding to
  a positive eigenvalue) we would have $\e^{\sigma r}\phi\in
  L^2(\Omega)$ for all $\sigma\geq 0$. Consequently what would remain
  to be shown for completing the induction argument is that
  super-exponentially decaying eigenstates vanish.

  Although we are not able in general to implement  the above scheme  for showing
  absence of positive eigenvalues we show a partial result which
  reduces the problem to resolvent estimates for sub-system type
  Hamiltonians. Moreover we do in fact implement the scheme for $N=2$
  under additional conditions.
  \begin{cond}\label{cond:smooth3}
    Suppose $N\geq 2$. For all $b\in \vB\setminus\{a_{\max}\}$ with
    $\Omega_b\subsetneq\bX^b$ the set
    $\Theta_b:=\bX^b\setminus\overline {\Omega_b}\neq\emptyset$ has
    smooth boundary $\partial \Theta_b=\partial \Omega_b$ and is
    strictly convex.
  \end{cond}

  For the notion of strict convexity used in this paper we refer to
  Appendix \ref{sec:Appendix}. Given Condition \ref{cond:smooth3}, by
  definition  if
  $\Omega_b\subsetneq\bX^b$, then $\dim \bX^b\ge 2$. With minor
  modifications we could have allowed  $\dim \bX^b=1$ in the
  definition of strict convexity and obtained the same results, however for convenience we
 prefer not to  do that. 

  The main result of this paper is the following.
  \begin{thm}\label{thm:gener-hard-core} Suppose $N=2$ and Conditions
    \ref{cond:smooth2} and \ref{cond:smooth3}. Suppose that for all
    $b\in \vB\setminus\{a_{\max}\}$ with $\Omega_b\subsetneq\bX^b$ the
    term $V^{(2)}_b=0$ while for all $b\in \vB\setminus\{a_{\max}\}$
    with $\Omega_b=\bX^b$
    \begin{equation}
      \label{eq:2k2bb}
      x^b\cdot \nabla V_{b}^{(2)}(x^b),\,(x^b\cdot \nabla)^2 V_{b}^{(2)}(x^b)\in \vC\big (H^1(\bX^b),H^1(\bX^b)^*\big ).
    \end{equation} Suppose that any eigenstate of $H$   vanishing at infinity
    must be zero (the unique continuation property). Then $H$ does not have positive eigenvalues.
  \end{thm}
  The unique continuation property is a well-studied subject, see for
  example \cite{Ge, JK, RS, Wu}. It is valid for a large class of
  potential singularities given connectedness of $\Omega$.
  \begin{corollary}
    \label{cor:gener-hard-core} For $N=2$ charged particles confined
    to the exterior of a bounded strictly convex obstacle
    $\Theta_1\subset \mathbb{R}^d$ containing $0$, $d\geq2$, the
    corresponding Hamiltonian $H$ given by \eqref{eq:Hami_1} with
    Coulomb interactions $V^{\rm ncl}_{j}(y)=q_jq^{\rm ncl}|y|^{-1}$
    and $V_{ij}^{\rm elec}(y)=q_iq_j|y|^{-1}$ does not have positive
    eigenvalues.
  \end{corollary} Note for Corollary \ref{cor:gener-hard-core} that
  indeed $\Omega=(\Omega_1)^2\setminus\{(x_1,x_2)\in
  (\R^{d})^2|x_1=x_2\}$ is a connected subset of $\R^{2d}$ for
  $d\geq2$   (which follows readily using that $\Omega_1\subset \R^{d}$ is
  connected) and that the version of the unique continuation property
  of \cite{RS} applies.

  Another result of this paper is the following statement in which a
  technical condition stated in Section \ref{sec:Reduction to
    high-energy hard-core sub-system resolvent bounds} enters.
  \begin{proposition}\label{prop:N-body} Suppose $N\geq 2$ and
    Conditions \ref{cond:smooth2} and
    \ref{cond:general-scheme}. Suppose $H$ does not have positive
    thresholds. Suppose that  any eigenstate of $H$ vanishing at
    infinity must be zero (the unique continuation
    property). Then $H$
    does not have positive eigenvalues.
  \end{proposition}
  By imposing the analogous version of Condition
  \ref{cond:general-scheme} for sub-Hamiltonians as well as the unique
  continuation property for these operators and for $H$ (in addition
  to Condition \ref{cond:smooth2}) we obtain that $H$ does not have
  thresholds nor positive eigenvalues, cf. the scheme discussed
  above. However since we are only able to verify Condition
  \ref{cond:general-scheme} for $N=2$ using Condition
  \ref{cond:smooth3} we need these restrictions in Theorem
  \ref{thm:gener-hard-core}. Nevertheless, since proving Condition
  \ref{cond:general-scheme} for higher $N$ under Condition
  \ref{cond:smooth2} possibly as well as under Condition
  \ref{cond:smooth3} could be a purely technical difficulty, we
  consider the statement of Proposition \ref{prop:N-body} in such  situation as a result of independent
  interest. We devote Section \ref{sec:Reduction to high-energy
    hard-core sub-system resolvent bounds} to the crucial step in the
  proof. Section \ref{sec:One-body high-energy hard-core resolvent
    bound} is devoted to the verification of Condition
  \ref{cond:general-scheme} for $N=2$. Supplementary material is given
  in Appendices \ref{sec:Appendixb} and \ref{sec:Appendix}.

\subsection{Geometric properties}\label{Geometric properties}
We complete this section by a brief discussion of some properties
related to Lemma \ref{lemma:vector field}, and we show an estimate
which may be viewed as a first step in a proof of (a hard-core version
of) Lemma \ref{lemma:Mourre1} (not to be elaborated on in this
paper). These properties will be important in Section
\ref{sec:Reduction to high-energy hard-core sub-system resolvent
  bounds}.
\subsubsection{Potential function}
Since $\omega_*$ is symmetric we can write
\begin{equation*}
  \omega=\nabla r^2/2. 
\end{equation*}
It will be important for us that the function $r=r(x)$ can be chosen
positive, smooth and convex, see the proof of \cite[Proposition
4.4]{De} (we remark that \cite{De} also uses the Graf construction
although with a different regularization procedure). From the
convexity of $r$ we learn that
\begin{subequations}
  \begin{equation}
    \label{eq:posi_term}
    \partial^r|\d r|^2\geq 0;\;\partial^rf=\i p^rf:=\nabla r \cdot \nabla f.
  \end{equation}

  We have a slight extension of part of \eqref{eq:boun_constr},
  cf. \cite[Lemma 4.3 f)]{De},
  \begin{equation}
    \label{eq:small_terma}
    \forall \alpha\in \N_0^{\dim \bX}\mand
    k\in \N_0: |\partial_x^\alpha(x\cdot
    \nabla)^k\parb{r^2-x^2}|\leq C_\alpha.
  \end{equation} In particular 
  we obtain yet another useful property
  \begin{equation}
    \label{eq:small_termb}
    \forall \alpha\in \N_0^{\dim \bX}: |\partial_x^\alpha\parb{ |\d r|^2-1}|\leq C_\alpha \inp{x}^{-2}.
  \end{equation} In fact letting $f=r^2-x^2$ the bounds \eqref{eq:small_termb}
  follow  from \eqref{eq:small_terma} and the identity 
  \begin{equation*}
    |\d r|^2-1=\frac {x\cdot \nabla f +4^{-1} |\d f|^2-f}{x^2+f}.
  \end{equation*} The rescaled $r$ reads
  \begin{equation*}
    r_R(x)=Rr(x/R),
  \end{equation*} so that $\omega_R=\nabla r_R^2/2$. Clearly the bounds
  \eqref{eq:posi_term}-\eqref{eq:small_termb} are also valid  for the
  rescaled $r$ (possibly with $R$-dependent constants). We also
  rescale the partition functions of Lemma \ref{lemma:vector field}
  $\tilde q_{a,R}(x):=\tilde q_{a}(x/R)$ and similarly for the
  ``quadratic'' partition functions
  \begin{equation*}
    q_b(x)=\tilde q_b(kx)\parb{\sum_c \tilde q_c(kx)^2}^{-1/2};\;k=r_1/r_2.
  \end{equation*} Using that
  \begin{equation*}
    \tilde q_c(x)\tilde q_b(kx)=0\text{ if }c\not
    \subset b,
  \end{equation*} 
\end{subequations} and Lemma \ref{lemma:vector field} \ref{item:4k} we
conclude that
\begin{equation}
  \label{eq:2deri}
  \omega_*(x)\geq\sum _b\pi_b q_b^2(x).
\end{equation}

\subsubsection{Commutator calculation} We calculate
\begin{equation}\label{eq:mourre comm}
  \i [H,A_R]=2p\omega_*(x/R)p-(4R^2)^{-1}\parb{\Delta^2
    r^2}(x/R)-2\omega_R\cdot \nabla V,
\end{equation} and using \eqref{eq:2deri} we thus deduce
\begin{align}\label{eq:basic_est}
  \begin{split}
    \i [H,A_R]&\geq 2\sum _bq_{b,R}\,p_b^2\,q_{b,R}+O\parb
    {R^{-2}}-2\omega_R\cdot \nabla V\\
    &= 2\sum _bq_{b,R}\,p_b^2\,q_{b,R}+O\parb
    {R^{-\min\{2,\varepsilon\}}}.
  \end{split}
\end{align}

  \subsubsection{More  notation}\label{subsec:Notation}
  We fix a non-negative $\chi\in C^\infty(\mathbb{R})$ with $0\leq
  \chi\leq 1$ and
  \begin{align*}
    \chi(t)=\left\{\begin{array}{ll}
        0&\mbox{ for } t\le 5/4, \\
        1 &\mbox{ for }t \ge 7/4.
      \end{array} \right.
  \end{align*}
  We shall frequently use the rescaled function
  \begin{align}
    \chi_\nu(t)=\chi(t/\nu),\quad \nu>0,
    \label{eq:11.7.11.5.14}
  \end{align} and the notation $\chi_\nu^+=\chi_\nu$ and
  $\chi_\nu^-=\bar \chi_\nu=1-\chi_\nu$.

  For any self-adjoint operator $T$ and state $\phi$ we write
  $\inp{T}_\phi=\langle \phi ,T\phi\rangle$ for the corresponding
  expectation value.

  \section{Reduction to high-energy hard-core sub-system resolvent
    bounds}\label{sec:Reduction to high-energy hard-core sub-system
    resolvent bounds} Under Condition \ref{cond:smooth2} we propose a
  scheme for showing
  \begin{equation}
    \label{eq:hard_part}
    (H-E)\phi=0,\;E>0,  \mand \forall \sigma\geq
    0:\;\e^{\sigma r}\phi\in \vH=L^2(\Omega)\Rightarrow \phi=0.
  \end{equation}
  Here and henceforth $r=r_R$ is the rescaled potential function. We
  suppress the dependence on $R$ which is fixed (large) from this
  point. The proposed method relies on the unique continuation
  property and
  certain high-energy hard-core sub-system type resolvent bounds. The
  latter are stated in Condition \ref{cond:general-scheme} given
  below. Whence we give the crucial step of  the proof  Proposition
  \ref{prop:N-body}.

\subsection{General scheme} 
For $\phi$ given as in \eqref{eq:hard_part} we let for any $\nu\ge 1$
and $\sigma\ge 0$
\begin{align}
  \phi_\sigma=\phi_{\sigma,\nu}:=\chi_\nu\e^{\sigma (r-4\nu)}\phi\in
  \vH;\;\chi_\nu=\chi_\nu(r).
  \label{eq:11.7.3.23.49c}
\end{align} Putting $H^\sigma=H-\tfrac{\sigma^2}{2}|\mathrm{d} r|^2$
we note that
\begin{equation}\label{eq:19c}
  \begin{split}
    (H^\sigma-E)\phi_\sigma ={}& -\mathrm{i}\sigma
    (\mathop{\mathrm{Re}}{}p^r)\phi_\sigma
    -\mathrm{i}\mathrm{e}^{\sigma (r-4\nu)} R(\nu)\phi,
  \end{split}
\end{equation} where $R(\nu)=\mathrm{i}[H_0,\chi_\nu]=\mathop{\mathrm{Re}}{}\parb{\chi_\nu'p^r}$.
Whence by undoing the commutator, cf.
Appendix~\ref{sec:Appendixb},
\begin{align}
  \inp{\i[H^\sigma,A]}_{\phi_\sigma} =-2\sigma
  \mathop{\mathrm{Re}}{}\inp{(\mathop{\mathrm{Re}} p^r)
    A}_{\phi_\sigma} -2\mathop{\mathrm{Re}}{}\inp{R(\nu)
    \mathrm{e}^{\sigma (r-4\nu)}A\chi_\nu\mathrm{e}^{\sigma
      (r-4\nu)}}_{\phi}.
  \label{eq:11.7.13.4.10c}
\end{align}
The first term of (\ref{eq:11.7.13.4.10c}) is computed
\begin{align}\label{eq:27first}
  \begin{split}
    &-2\sigma \mathop{\mathrm{Re}}{}((\mathop{\mathrm{Re}} p^r) A)\\&
    =-\sigma(\mathop{\mathrm{Re}} p^r) (2r\mathop{\mathrm{Re}}{}p^r
    -\mathrm{i}|\mathrm{d} r|^2)
    +\text{h.c.}\\
    &= -4\sigma(\mathop{\mathrm{Re}} p^r)
    r\mathop{\mathrm{Re}}{}p^r+\sigma(\partial^r|\mathrm{d} r|^2).
  \end{split}
\end{align}
As for the second term we estimate (recall the notation $\bar
\chi_\nu=1-\chi_{\nu}$)
\begin{align*}
  &-2\mathop{\mathrm{Re}}{}\inp{R(\nu) \mathrm{e}^{\sigma (r-4\nu)}
    A\chi_\nu\mathrm{e}^{\sigma (r-4\nu)}}_{\phi}\\
  &\le \|\mathrm{e}^{\sigma (r-4\nu)}R(\nu)\phi\|^2+
  \|\bar\chi_{2\nu}A\chi_\nu\mathrm{e}^{\sigma (r-4\nu)}\phi\|^2\\
  &\le \bigl\{\|\chi_\nu'\mathrm{e}^{\sigma (r-4\nu)}p^r\phi\|
  +\tfrac{1}{2}\|(\chi_\nu''|\mathrm{d} r|^2+\chi'_\nu(\Delta r))
  \mathrm{e}^{\sigma (r-4\nu)}\phi\|\bigr\}^2\\
  &\phantom{\le
    {}}+\bigl\{\|2r\bar\chi_{2\nu}\chi_\nu\mathrm{e}^{\sigma
    (r-4\nu)}p^r \phi\|+\|\bar\chi_{2\nu}(2r|\mathrm{d}
  r|^2\chi_\nu'+2\sigma r\chi_\nu|\mathrm{d}r|^2
  +\tfrac{1}{2}(\Delta r^2)\chi_\nu) \mathrm{e}^{\sigma (r-4\nu)} \phi\|\bigr\}^2\\
  &\le C\nu^2\|\chi_{\nu/2}|p \phi|\|^2+C\nu^2\langle
  \sigma\rangle^2\|\phi\|^2\\
  &\le C\nu^2\inp{p^2}_ \phi +C\nu^2\langle \sigma\rangle^2\|\phi\|^2.
\end{align*}
Note that $C>0$ does not depend on $\nu$ or $\sigma$ because $r\leq
2\nu$ on $\mathop{\mathrm{supp}} \chi_\nu'$.  Using the relative 
$\epsilon$-smallness of the potential we have for some $C>0$
\begin{equation}\label{eq:bnd_potentials}
  \inp*{p^2}_ \phi \leq \inp{4H+C}_ \phi=(4E+C)\|\phi\|^2,
\end{equation} 
and we deduce that
\begin{align}\label{eq:28second}
  -2\mathop{\mathrm{Re}}{}\inp{R(\nu) \mathrm{e}^{\sigma
      (r-4\nu)}A\chi_\nu \mathrm{e}^{\sigma (r-4\nu)}}_{\phi} \le
  C\nu^2\langle\sigma\rangle^2\|\phi\|^2.
\end{align}

On the other hand doing the commutator, cf. \eqref{eq:mourre comm},
and then using \eqref{eq:posi_term} and \eqref{eq:basic_est} we obtain
that
\begin{align}
  \begin{split}
    \inp{\i[H^\sigma,A]}_{\phi_\sigma} &\geq -\sigma^2 \Im \inp {A|\d
      r|^2}_{\phi_\sigma}
    +2\sum _b\inp{p_b^2}_{q_{b,R}\phi_\sigma}-o(R^0)\|\phi_\sigma\|^2 \\
    &= \sigma^2 \inp {r\partial^r|\d r|^2}_{\phi_\sigma}
    +2\sum _b\inp{p_b^2}_{q_{b,R}\phi_\sigma}-o(R^0)\|\phi_\sigma\|^2 \\
    &\geq \sigma \inp {\partial^r|\d r|^2}_{\phi_\sigma} +2\sum _b\inp
    {p_b^2}_{q_{b,R}\phi_\sigma}-o(R^0) \|\phi_\sigma\|^2
  \end{split}
  \label{eq:11.7.13.4.10cc}
\end{align}

We combine \eqref{eq:11.7.13.4.10c}--\eqref{eq:11.7.13.4.10cc} and
conclude that
\
\begin{align}
  \begin{split}
    &C\nu^2\langle\sigma\rangle^2\|\phi\|^2+o(R^0)\|\phi_\sigma\|^2  \\
    &\geq 4\sigma\inp { r}_{(\Re{p^r})\phi_\sigma} +2\sum _b\inp
    {p_b^2}_{q_{b,R}\phi_\sigma}.
  \end{split}
  \label{eq:11.7.13.4.10ccc}
\end{align}

We aim at deriving some useful positivity from the second term of
\eqref{eq:11.7.13.4.10ccc} to the right. For that let us for $b\in
\vA$ introduce
\begin{align}
  \label{eq:Ham_dis}
  \widetilde H_b&=
  \widetilde H^b+\tilde p_b^2\chi^-_{\sigma^2/2}\parb{ \tilde p_b^2};\\
  s(x)&=\chi^+_{\nu/2}(r)(|\d r|-1)/\sqrt{2}+1/\sqrt{2},\nonumber\\
  \tilde p_b^2&=\tfrac12 s(x)^{-1}p_b^2s(x)^{-1},\nonumber\\
  \widetilde H^b&=s(x)^{-1}H^bs(x)^{-1}.\nonumber
\end{align} Here we suppressed the dependence of $\widetilde H_b
$ on 
the parameter $R$ (through $r$, and considered as fixed) as well as
the dependence on $\nu$ and $\sigma$. The latter parameters will be
considered as independent large parameters (at the end we fix $\nu$
large and let $\sigma\to \infty$). The operator $\widetilde
H_b-\sigma^{2}$ should be thought of as an effective approximation to 
\begin{align*}
  2|\d r|^{-1}\parb{H_b-\tfrac{\sigma^2}{2}|\mathrm{d} r|^2} |\d
  r|^{-1}\approx 2H_b-\sigma^2=2H^b+p_b^2-\sigma^2.
\end{align*}
Let us here note the following consequence of \eqref{eq:small_termb}
\begin{equation}
  \label{eq:der_s}
  \forall \alpha\in \N_0^{\dim \bX}: |\partial^\alpha_x\parb{s(x)-1/\sqrt{2}}|\leq C_\alpha\nu^{-2}.
\end{equation}

The definitions \eqref{eq:Ham_dis} are accompanied by the following
specification of domains: For $b\in \vA$ we define
\begin{align*}
  \vH_b=L^2(\Omega^b) \otimes L^2(\bX_b)=L^2(\Omega^b+\bX_b),
\end{align*} and note that
\begin{align*}
  \vH_b\subset L^2(\bX^b) \otimes L^2(\bX_b)=L^2(\bX).
\end{align*}
The operator $\tilde p_b^2$ is an operator on $L^2(\bX)$ with
$Q(\tilde p_b^2)=Q(p_b^2)=L^2(\bX^b) \otimes H^1(\bX_b)$ and
$\vD(\tilde p_b^2)=\vD(p_b^2)=L^2(\bX^b) \otimes H^2(\bX_b)$. However
since it is multiplicative in the $x^b$ variable the space $\vH_b$ is
an invariant subspace, in fact $\vD(\tilde p_b^2)\cap
\vH_b=L^2(\Omega^b) \otimes H^2(\bX_b)$. Whence clearly the first
term of \eqref{eq:Ham_dis} is a bounded operator on $\vH_b$ (with the
norm bound $\tfrac 78 \sigma^2$). The second term of
\eqref{eq:Ham_dis}, $\widetilde H^b$, is also an operator on
$\vH_b$. We specify its form domain to be
$L^2(\bX_b,H^1_0(\Omega^b);\d x_b)$.  The corresponding quadratic
form is closed. We conclude the  same for $\widetilde H_b$, and
consequently $\widetilde H_b$ is a well-defined operator on $\vH_b$
with
\begin{align*}
  Q(\widetilde H_b)=L^2(\bX_b,H^1_0(\Omega^b);\d x_b)\subset \vH_b.
\end{align*}
For later applications let us note the facts that $\vD(\tilde
p_c^2)\supset \vD(\tilde p_b^2)$ and $\vH_b\supset \vH_c$ for all
$c\supset b$ (the latter embedding is due to the relation $\Omega^b+
\bX_b\supset \Omega^c+ \bX_c$).

We introduce a technical condition for the operators introduced in
\eqref{eq:Ham_dis}.
\begin{cond}
  \label{cond:general-scheme} For all $b\neq a_{\max}$ the following
  bound holds uniformly in all large $\sigma,\nu>1$, $\epsilon\in
  (0,1]$ and reals $\lambda$ near $1$:
  \begin{equation}
    \label{eq:Basic_assump}
    \|\delta_\epsilon(\widetilde H_{b}/\sigma^2-\lambda)\|_{\vB(B(|x^b|),B(|x^b|)^*)}\leq
    C \sigma.
  \end{equation} Here, by 
  definition, for any self-adjoint operator $T$
  \begin{align*}
    \delta_\epsilon(T)=\pi^{-1}\Im ( T-\i \epsilon )^{-1}.
  \end{align*} The space $B(\cdot)$ is a Besov space, see Subsection
  \ref{subsec: Abstract Besov spaces} for the abstract
  definition. Note that \eqref{eq:Basic_assump} is trivially fulfilled
  for $b= a_{\min}$ (by the spectral theorem). We derive the bounds
  for $N=2$ in Section \ref{sec:One-body high-energy hard-core
    resolvent bound} under the additional regularity conditions on the
  obstacles and potentials stated in Theorem \ref{thm:gener-hard-core}. Note that for $N=2$ and $b\notin\{a_{\min},
  a_{\max}\}$ only $b'=b $ obeys $a_{\min}\neq b'\subset b$ and hence
  for such $b$ \eqref{eq:Basic_assump} is an effective high energy
  bound for a bounded obstacle (hence one-body type). More generally
  we prove \eqref{eq:Basic_assump} for $b$ with $\# b=N$ under the
  additional regularity conditions for $\Omega_b$ and $V_b$. High energy
  resolvent bounds are studied previously in the literature, see for
  example \cite{Je,Vo1,Vo2, RT}.
\end{cond}

Let us also introduce $\tilde\phi_\sigma=s(x) \phi_\sigma$. We
estimate for $b\neq a_{\max}$, $k_1>0$ (can be fixed arbitrarily) and
all large $\sigma>1$
\begin{align}
  \label{eq:fun_est}
  \begin{split}
    \tfrac12 \inp {p_b^2}_{q_{b,R}\phi_\sigma}&= \inp { \tilde
      p_b^2}_{q_{b,R}\tilde \phi_\sigma}\\
    &\geq \inp { \tilde
      p_b^2\chi^+_{k_1\sigma}\parb{ \tilde p_b^2}}_{q_{b,R}\tilde \phi_\sigma}\\
    &\geq k_1\sigma \parb{\|q_{b,R}\tilde
      \phi_\sigma\|^2-\inp{\chi^-_{k_1\sigma}\parb{ \tilde
          p_b^2}}_{q_{b,R}\tilde \phi_\sigma}}\\
    &\geq k_1\sigma \parb{\|q_{b,R}\tilde
      \phi_\sigma\|^2-\inp{\chi^-_{\sigma^2/8}\parb{ \tilde
          p_b^2}}_{q_{b,R}\tilde \phi_\sigma}}.
  \end{split}
\end{align} The contribution to \eqref{eq:11.7.13.4.10ccc} from the
first term to the right in \eqref{eq:fun_est} amounts (for $\nu\geq R
r_2$) to the positive term $ 4k_1\sigma \|\tilde \phi_\sigma\|^2$, and
it remains to estimate the contribution from the second term to the
right. Now up to a term of order $O(\sigma^{-2})$ better it is given
by summing the expressions $-4k_1\sigma\Re
\inp{\chi^-_{\sigma^2/8}\parb{ \tilde p_b^2}q_{b,R}^2}_{\tilde
  \phi_\sigma}$.

We write for $b\neq a_{\max}$ and $R_1\geq Rr_2/r_1$
\begin{align*}
  q_{b,R}^2=q_{b,R}^2\sum_{b_1\supset {b}}
  q_{b_1,R_1}^2=q_{b,R}^2q_{b,R_1}^2+q_{b,R}^2\sum_{b_1\supsetneq {
      b}} q_{b_1,R_1}^2.
\end{align*} Actually we shall later need $R_1>>R$. We repeat the
expansion by writing for $R_2>>R_1$ and $b_1\supsetneq b$
\begin{align*}
  q_{b,R}^2q_{b_1,R_1}^2=q_{b,R}^2q_{b_1,R_1}^2
  q_{b_1,R_2}^2+\sum_{b_2\supsetneq {b_1}} q_{b,R}^2
  q_{b_1,R_1}^2q_{b_2,R_2}^2.
\end{align*} Upon further iteration the procedure stops for each
branch after say $n$ times when necessarily $b_n=a_{\max}$ ($n$ is at
most $N$). Whence the only non-trivial terms to examine have the form
\begin{align*}
  q_{b,R}^2q_{b,R_1}^2\text{ or }q_{b,R}^2\prod_{1\leq j\leq
    m}q_{b_j,R_j}^2 q_{b_m,R_{m+1}}^2,
\end{align*}
where $m\le n-1$,  $b\subsetneq b_1\cdots \subsetneq b_m\subsetneq
a_{\max}$ and $R<<R_1\cdots <<R_m<<R_{m+1}$. As the reader will see
these constraints are needed later, see the verification of
\eqref{eq:11.7.13.4.10cccd}.  Moreover we shall need the constraint
$\nu\geq R_{m+1}r_2$. Introducing the notation $b_0=b$ and $R_0=R$ in
either case the form is then $q_{b_m,R_m}^2 q_{b_m,R_{m+1}}^2$ times a
bounded factor $Q_m^2$, in fact $|Q_m(x)|\leq 1$.  We decompose
\begin{align*}
  &\Re\parb{ \chi^-_{\sigma^2/8}\parb{ \tilde p_b^2}q_{b,R}^2}\\
  &=\Re\parb{ \sum \chi^-_{\sigma^2/8}\parb{ \tilde p_b^2}q_{b_m,R_m}^2q_{b_m,R_{m+1}}^2Q_m^2}+\Re\parb{ \sum \chi^-_{\sigma^2/8}\parb{ \tilde p_b^2}q_{b_m,R_m}^2q_{a_{\max},R_{m+1}}^2Q_m^2}\\
  &=\sum q_{b_m,R_{m+1}}q_{b_m,R_m}Q_m\chi^-_{\sigma^2/8}\parb{ \tilde
    p_b^2}Q_mq_{b_m,R_m} q_{b_m,R_{m+1}}+\text{ remainder}.
\end{align*} Here the remainder is the sum of terms either
$O(\sigma^{-2})$ better than ``the good term'' $ 4k_1\sigma \|\tilde
\phi_\sigma\|^2$ we derived from \eqref{eq:fun_est} or being expressed
by factors of $q_{a_{\max},R_{m+1}}$. Whence (using $\nu\geq
R_{m+1}r_2$) the remainder conforms with \eqref{eq:11.7.13.4.10ccc}.

Next on both sides of the factor $\chi^-_{\sigma^2/8}\parb{ \tilde
  p_b^2}$ in the summation to the right we insert
\begin{align*}
  I=\chi^-_{\sigma^2/4}\parb{ \tilde
    p_{b_m}^2}+\chi^+_{\sigma^2/4}\parb{ \tilde p_{b_m}^2}.
\end{align*} This yields four times as many terms. The contribution
from the terms with two factors of $\chi^-_{\sigma^2/4}$ is then
estimated as
\begin{align}
  &\sum q_{b_m,R_{m+1}}q_{b_m,R_m}Q_m\chi^-_{\sigma^2/4}\parb{ \tilde
    p_{b_m}^2}\chi^-_{\sigma^2/8}\parb{ \tilde
    p_b^2}\chi^-_{\sigma^2/4}\parb{ \tilde p_{b_m}^2}Q_mq_{b_m,R_m}
  q_{b_m,R_{m+1}}\nonumber\\
  &\leq \Re \parbb{\sum q_{b_m,R_{m+1}}^2\chi^-_{\sigma^2/4}\parb{
      \tilde p_{b_m}^2}^2q_{b_m,R_m}^2
    Q_m^2}+O(\sigma^{-2}).\label{eq:lead_term}
\end{align} We take a closer look at the first term later. We first
consider the contributions from
\begin{align*}
  2\Re \parbb{\chi^+_{\sigma^2/4}\parb{ \tilde
      p_{b_m}^2}\chi^-_{\sigma^2/8}\parb{ \tilde
      p_b^2}\chi^-_{\sigma^2/4}\parb{ \tilde
      p_{b_m}^2}}+\chi^+_{\sigma^2/4}\parb{ \tilde
    p_{b_m}^2}\chi^-_{\sigma^2/8}\parb{ \tilde
    p_b^2}\chi^+_{\sigma^2/4}\parb{ \tilde p_{b_m}^2}.
\end{align*} For this purpose let us for $\psi\in L^2(\bX)$ introduce
\begin{align*}
  \psi_\sigma=\chi^+_{\sigma^2/4}\parb{ \tilde
    p_{b_m}^2}\chi^-_{\sigma^2/8}\parb{ \tilde p_b^2}\psi\mand
  \widetilde \psi_\sigma=\chi^-_{\sigma^2/8}\parb{ \tilde
    p_b^2}\chi^+_{\sigma^2/4}\parb{ \tilde p_{b_m}^2}\psi.
\end{align*}
We apply (see below for a proof)
\begin{subequations}
  \begin{align}\label{eq:comm1}
    \|[\chi^+_{\sigma^2/4}\parb{ \tilde
      p_{b_m}^2},\chi^-_{\sigma^2/8}\parb{ \tilde p_b^2}]\|\leq
    C\tfrac1{\sigma\nu},
  \end{align} yielding
  \begin{align*}
    \|\psi_\sigma-\widetilde \psi_\sigma\|\leq
    C\tfrac1{\sigma\nu}\|\psi\|,
  \end{align*} and whence with the operator monotone function
  $f(t)=(t-1)/(1+t)$ 
  \begin{align*}
    \|\widetilde \psi_\sigma\|^2&\leq 2\|\psi_\sigma\|^2+C\tfrac1{\sigma^2\nu^2}\|\psi\|^2\\
    &\leq 2f(\tfrac54) ^{-1}\inp{f(\tfrac{4}{\sigma^2}\tilde
      p_{b_m}^2)}_{\psi_\sigma}+C\tfrac1{\sigma^2\nu^2}\|\psi\|^2\\
    &\leq 2f(\tfrac54) ^{-1}\inp{f(\tfrac{4}{\sigma^2}\tilde
      p_{b}^2)}_{\psi_\sigma}+C\tfrac1{\sigma^2\nu^2}\|\psi\|^2.
  \end{align*}

  Obviously it follows from \eqref{eq:comm1} that
  \begin{align}\label{eq:comm2}
    \|f(\tfrac{4}{\sigma^2}\tilde p_{b}^2)[\chi^+_{\sigma^2/4}\parb{
      \tilde p_{b_m}^2},\chi^-_{\sigma^2/8}\parb{ \tilde p_b^2}]\|\leq
    C\tfrac1{\sigma\nu},
  \end{align} and we can ``reverse the commutation''
  \begin{align*}
    2f(\tfrac54) ^{-1}&\inp{f(\tfrac{4}{\sigma^2}\tilde
      p_{b}^2)}_{\psi_\sigma}\\&\leq 2f(\tfrac54)
    ^{-1}\inp{f(\tfrac{4}{\sigma^2}\tilde p_{b}^2)}_{\widetilde
      \psi_\sigma}+C\tfrac
    1{\sigma\nu}\|\psi\|\,\|\psi_\sigma\|+C\tfrac
    1{\sigma\nu}\|\psi\|\,\|\widetilde \psi_\sigma\|\\&\leq  2f(\tfrac54) ^{-1}f(\tfrac78)  \|\widetilde \psi_\sigma\|^2+\|\psi_\sigma\|^2+\epsilon\|\widetilde \psi_\sigma\|^2+C_\epsilon\tfrac1{\sigma^2\nu^2}\|\psi\|^2\\
    &\leq
    \|\psi_\sigma\|^2+C_\epsilon\tfrac1{\sigma^2\nu^2}\|\psi\|^2;
  \end{align*} here we took $\epsilon=2f(\tfrac54) ^{-1}|f(\tfrac78)|
  $, for example. By combining with the previous estimation we find
  \begin{align*}
    \|\psi_\sigma\|^2\leq C\tfrac1{\sigma^2\nu^2}\|\psi\|^2,
  \end{align*}
  and then in turn
  \begin{align*}
    \|\psi_\sigma\|,\,\|\widetilde\psi_\sigma\|\leq
    C\tfrac1{\sigma\nu}\|\psi\|.
  \end{align*} We conclude that indeed, due to errors of the form
  $O(\nu^{-1}\sigma^{-1})+O(\sigma^{-2})$, we need to examine the
  first term of \eqref{eq:lead_term} only.

  \subsubsection{Verification of \eqref{eq:comm1}} Introduce
  $P_m=\sigma^{-2}\tilde p_{b_m}^2$ and $P=\sigma^{-2}\tilde p_b^2$.
  We show the slightly stronger bound
  \begin{align}\label{eq:comm3}
    \|[\chi^+_{1/4}(P_m),\chi^-_{1/8}(P)]\|=\|[\chi^-_{1/4}(P_m),\chi^-_{1/8}(P)]\|\leq
    C\tfrac1{\sigma\nu^2}.
  \end{align}
  Since $P_m,P\geq 0$ we can truncate $\chi^-_\nu$, $\nu=1/4, 1/8$, at
  the negative half-axis to become functions $\chi_1,\chi_2$ in
  $C_{\c}^\infty(\R)$ and invoke the standard representation for a
  self-adjoint operator $T$ and such function $\chi$
\end{subequations}
\begin{equation}\label{eq:rep1}
  \chi(T)=\int_{\mathbb{C}}(T-z)^{-1}\,\mathrm{d}\mu(z),\;\mathrm{d}\mu(z)=-\frac{1}{2\pi
    \mathrm{i}}\bar{\partial}\tilde{\chi}(z)\mathrm{d}z\mathrm{d}\bar{z}, 
\end{equation} where 
we have used  an almost analytic extension $\tilde{\chi}\in
C^\infty_{\mathrm{c}}(\mathbb{C})$, 
i.e.
\begin{align*}
  \tilde{\chi}(t)=\chi(t)\mfor t\in \mathbb{R},&&
  |\bar{\partial}\tilde{\chi}(z)|\le
  C_k|\mathop{\mathrm{Im}}z|^k;\;k\in\N.
\end{align*}

Whence
\begin{subequations}
  \begin{align}\label{eq:rep2}
    \chi^-_{1/4}(P_m)&=\int_{\mathbb{C}}(P_m-z_1)^{-1}\,\mathrm{d}\mu_1(z_1),\\
    \chi^-_{1/8}(P)&=\int_{\mathbb{C}}(P-z_2)^{-1}\,\mathrm{d}\mu_2(z_2).\label{eq:rep2b}
  \end{align}
\end{subequations}

Using \eqref{eq:rep2}, \eqref{eq:rep2b} and the domain relation
$\vD(P_m)\supset\vD(P)$ we represent 
\begin{align*}
  [\chi^-_{1/4}(P_m),\chi^-_{1/8}(P)]=\int_{\mathbb{C}}\int_{\mathbb{C}}&(P_m-z_1)^{-1}(P-z_2)^{-1}\\&[P_m,P](P-z_2)^{-1}(P_m-z_1)^{-1}\,\mathrm{d}\mu_2(z_2)\mathrm{d}\mu_1(z_1).
\end{align*}
Next we note the elementary bounds
\begin{subequations}
  \begin{align}
    \label{eq:ele_bound1}
    \|(P_m-z)^{-1}\|&\leq \tfrac{1}{|\Im z|},\\
    \|\inp{P}(P-z)^{-1}\|&\leq C\tfrac{|z|+1}{|\Im
      z|}.\label{eq:ele_bound3}
  \end{align}

  Using \eqref{eq:der_s} we compute
  \begin{equation}
    \label{eq:commb}
    \|\inp{P}^{-1}  [P_m,P]\inp{P}^{-1}\|\leq C\tfrac1{\sigma\nu^2}.
  \end{equation}
\end{subequations}

Finally applying \eqref{eq:ele_bound1}--\eqref{eq:commb} to the double
integral we obtain the bound
\begin{align*}
  \cdots \leq
  C_1\tfrac1{\sigma\nu^2}&\int_{\mathbb{C}}\int_{\mathbb{C}} |\Im
  z_1|^{-2}\tfrac{(|z_2|+1)^{2}}{|\Im
    z_2|^2}\,|\mathrm{d}\mu_2(z_2)|\,|\mathrm{d}\mu_1(z_1)|=C_2\tfrac1{\sigma\nu^2},
\end{align*} and we have shown \eqref{eq:comm3}.

\subsubsection{Localization for first term of \eqref{eq:lead_term}} We
decompose (for $k_2>0$ to be fixed later, big)
\begin{align*}
  q_{b_m,R_{m+1}}^2\chi^-_{\sigma^2/4}\parb{ \tilde
    p_{b_m}^2}^2q_{b_m,R_m}^2
  Q_m^2&=q_{b_m,R_{m+1}}^2\chi^-_{\sigma^2/4}\parb{ \tilde
    p_{b_m}^2}(\chi^-+\chi^+)q_{b_m,R_m}^2
  Q_m^2;\\
  \chi^-&:=\chi^-_{(k_2R_m\sigma)^{-1}}(|\widetilde H_{b_m}/\sigma^2-1|),\\
  \chi^+&:=\chi^+_{(k_2R_m\sigma)^{-1}}(|\widetilde
  H_{b_m}/\sigma^2-1|).
\end{align*} To treat the contribution from $\chi^-$ we write
$q_{b_m,R_m}^2=\chi^-_{r_2R_m}(|x^{b_m}|)q_{b_m,R_m}^2$. Now imposing
Condition \ref{cond:general-scheme} and applying
\eqref{eq:Basic_assump} with $b=b_m$ we get using
\begin{align*}
  \|\chi^-_{r_2R_m}(|x^{b_m}|)\|_{\vB(\vH_{b_m},B(|x^{b_m}|))}\leq
  C\sqrt{r_2R_m}
\end{align*}
that
\begin{equation}
  \label{eq:bound_uncer}
  \|\chi^-_{r_2R_m}(|x^{b_m}|)(\chi^-)^2\chi^-_{r_2R_m}(|x^{b_m}|)\|_{\vB(\vH_{b_m})}\leq C_1 /k_2.
\end{equation} Here we used the general bound for $S$ bounded and $T$ self-adjoint
\begin{align*}
  \|S^*g(T)S\|\leq \|g\|_{L^1}\sup_{\lambda\in \supp g,\epsilon\in
    (0,1]} \|S^*\delta_\epsilon(T-\lambda)S\|,
\end{align*}
cf. Stone's formula \cite{RS}. We fix $k_2$ such that
$(\#\vA)^{N+1}\sqrt{C_1/k_2}\leq 1/2$, saving ``the good term''
$2k_1\sigma \|\tilde \phi_\sigma\|^2$ in the previous bound
\eqref{eq:11.7.13.4.10ccc}.

\subsubsection{Completion of proof of \eqref{eq:hard_part}} 
We need to examine the contribution from $ \chi^+$ to
\eqref{eq:11.7.13.4.10ccc}. We write
\begin{align*}
  q_{b_m,R_{m+1}}^2\chi^-_{\sigma^2/4}\parb{ \tilde p_{b_m}^2}^2
  \chi^+=k_2R_m\sigma q_{b_m,R_{m+1}}^2\chi^-_{\sigma^2/4}\parb{
    \tilde p_{b_m}^2}^2(\widetilde H_{b_m}/\sigma^2-1) \widetilde Q_m,
\end{align*} where $\widetilde Q_m=\widetilde Q_m(\widetilde
H_{b_m}/\sigma^2)$ is bounded with norm at most $1$. Taking
expectation in $\tilde \phi_\sigma$ and using the Cauchy Schwarz
inequality it suffices to bound
\begin{align}
  \begin{split}
    &4k_1\sigma\sum k_2R_m\sigma^{-1}\|(\widetilde H_{b_m}-\sigma^2)
    \chi^-_{\sigma^2/4}\parb{ \tilde p_{b_m}^2}^2
    q_{b_m,R_{m+1}}^2\tilde \phi_\sigma\|\,\|\tilde \phi_\sigma\|\\
    &\leq k_1\sigma \|\tilde \phi_\sigma\|^2+4\sigma\inp {
      r}_{(\Re{p^r})\phi_\sigma}+C\parb{\nu^2\langle\sigma\rangle^2\|\phi\|^2+\|\phi_\sigma\|^2}.
  \end{split}
  \label{eq:11.7.13.4.10cccd}
\end{align}

With \eqref{eq:11.7.13.4.10ccc} this yields
\begin{align*}
  C\parb{\nu^2\langle\sigma\rangle^2\|\phi\|^2+\|\phi_\sigma\|^2 }\geq
  k_1\sigma\parbb{1-O\parb{\nu^{-1}\sigma^{-1}}-O\parb{\sigma^{-2}}}
  \|\tilde \phi_\sigma\|^2,
\end{align*}
and we learn by letting $\sigma\to\infty$ that
$\chi_{4\nu}(r)\phi\equiv 0$ (for $\nu$ large), and then in turn from
the unique continuation property that $\phi= 0$.

\subsubsection{Mapping properties} As a preparation for proving
\eqref{eq:11.7.13.4.10cccd} let us note the following mapping
properties of $S=\chi^-_{\sigma^2/4}\parb{ \tilde p_{a}^2}^2$ and $T=
q_{a,R_{m+1}}^2$ entering in \eqref{eq:11.7.13.4.10cccd} with
$a=b_m$. Recall that the form domain $Q\parb{\widetilde H_{a}}$ of
$\widetilde H_{a}$ is $L^2(\bX_a,H^1_0(\Omega^a);\d x_a)$.
\begin{subequations}
  \begin{align}
    \label{eq:1map}S\in \vB\parb{ Q\parb{\widetilde H_{a}}, H^1_0(\Omega^a+ \bX_a)},\\
    T\in \vB\parb{H^1_0(\Omega^a+ \bX_a),H^1_0(\Omega)},
    \label{eq:2map}\\
    TS\in \vB\parb{Q\parb{\widetilde H_{a}},H^1_0(\Omega)},
    \label{eq:3map}\\
    TS\in \vB\parb{Q\parb{\widetilde H_{a}}}.
    \label{eq:4map}
  \end{align}
\end{subequations} For \eqref{eq:1map} we can use \eqref{eq:rep1} to
represent $S=\chi\parb{\tilde p_{a}^2}$, apply the integral to a
simple tensor $\psi^a\otimes \psi_a$ and then calculate derivatives of
the resulting expression (not to be elaborated on). Clearly $S$ is a
smoothing operator in the $x_a$ variable yielding the improved
smoothness. Also we note that since $S$ is multiplicative in the $x^a$
variable it preserves the support in this variable of elements of an
approximating sequence. We obtain that indeed $S\psi_a\otimes
\psi^a\in H^1_0(\Omega^a+ \bX_a)$ with
\begin{align*}
  \| S\psi^a\otimes \psi_a\|_1\leq C\|\psi^a\otimes
  \psi_a\|_{Q\parb{\widetilde H_{a}}}.
\end{align*}
This bound extends to finite sums of simple tensors (by the same
arguments) and hence \eqref{eq:1map} follows by density and
continuity. As for \eqref{eq:2map} we use that
\begin{align*}
  \supp\parb{q_{a,R_{m+1}}^2}\cap\parb{\Omega^a+ \bX_a}\subset
  \Omega.
\end{align*} Clearly \eqref{eq:3map} follows from \eqref{eq:1map} and
\eqref{eq:2map}, while in turn \eqref{eq:4map} follows from
\eqref{eq:3map} and the inclusion $\Omega\subset \Omega^a+ \bX_a$
implying that $H^1_0(\Omega)$ is continuously embedded in
$H^1_0(\Omega^a+ \bX_a)$ and therefore in $Q\parb{\widetilde H_{a}}$.

\subsubsection {Proof of \eqref{eq:11.7.13.4.10cccd}} We consider the
vector $(\widetilde H_{a}-\sigma^2) \chi^-_{\sigma^2/4}\parb{ \tilde
  p_{a}^2}^2 q_{a,R_{m+1}}^2\tilde \phi_\sigma$, $a=b_m$, in
\eqref{eq:11.7.13.4.10cccd} as an element of the dual space of the
form domain $Q\parb{\widetilde H_{a}}$, that is in
$L^2(\bX_a,H^1_0(\Omega^a)^*;\d x_a)$. As a part of
\eqref{eq:11.7.13.4.10cccd} we must show that indeed it belongs to
$\vH_a$. This will follow from \eqref{eq:3map} and the calculations
below. We rewrite, using that $\tilde \phi_\sigma\in Q\parb{\widetilde
  H_{a}}$, $\widetilde H_{a}\tilde \phi_\sigma\in Q\parb{\widetilde
  H_{a}}^*$ and \eqref{eq:4map},
\begin{align*}
  (\widetilde H_{a}-\sigma^2) \chi^-_{\sigma^2/4}\parb{ \tilde
    p_{a}^2}^2 q_{a,R_{m+1}}^2\tilde
  \phi_\sigma&=\chi^-_{\sigma^2/4}\parb{ \tilde p_{a}^2}^2
  q_{a,R_{m+1}}^2 (\widetilde
  H_{a}-\sigma^2) \tilde \phi_\sigma+T_{\rm com}\\
  T_{\rm com}&=[\widetilde H_{a}, \chi^-_{\sigma^2/4}\parb{ \tilde
    p_{a}^2}^2q_{a,R_{m+1}}^2]\tilde \phi_\sigma,
\end{align*} and then
\begin{align*}
  (\widetilde H_{a}-\sigma^2) \tilde \phi_\sigma&=
  s(x)^{-1}\parb{\tfrac12 p_{a}^2s(x)^{-1}\chi^-_{\sigma^2/2}\parb{
      \tilde p_{a}^2}s(x)+H^{a}-\tfrac{\sigma^2}{2}|\mathrm{d}
    r|^2}\phi_\sigma\\
  &=s(x)^{-1}\parb{(H^\sigma-E+\mathrm{i}\sigma
    (\mathop{\mathrm{Re}}{}p^r))\phi_\sigma+\mathrm{i}\mathrm{e}^{\sigma
      (r-4\nu)} R(\nu)\phi}+T_1+T_2+T_3\\
  T_1&=-\tilde p_{a}^2\chi^+_{\sigma^2/2}\parb{
    \tilde p_{a}^2}\tilde \phi_\sigma,\\
  T_2&=s(x)^{-1}\parb{E-I_{a}-\mathrm{i}\sigma
    (\mathop{\mathrm{Re}}{}p^r)}\phi_\sigma\\
  T_3&=-s(x)^{-1}\mathrm{i}\mathrm{e}^{\sigma (r-4\nu)} R(\nu)\phi.
\end{align*} Due to \eqref{eq:19c} (and \eqref{eq:3map}) we need to
estimate the contributions from $T_1$--$T_3$ and $T_{\rm com}$ only.

As for $T_1$ we note that
$(\chi^-_{\sigma^2/4})^2\chi^+_{\sigma^2/2}=0$. Whence by commutation
\begin{align*}
  \|\chi^-_{\sigma^2/4}\parb{ \tilde p_{a}^2}^2
  q_{a,R_{m+1}}^2T_1\|\leq C \tfrac{\sigma}{R_{m+1}}\|\tilde
  \phi_\sigma\|,
\end{align*} which agrees with \eqref{eq:11.7.13.4.10cccd} provided
$R_{m+1}>> R_m$.

We estimate
\begin{align*}
  \|\chi^-_{\sigma^2/4}\parb{ \tilde
    p_{a}^2}^2q_{a,R_{m+1}}^2T_2\|^2&\leq
  C\|\phi_\sigma\|^2+2\sigma^2\|(\mathop{\mathrm{Re}}{}p^r)\phi_\sigma\|^2\\
  &\leq (C+O(\sigma^2/\nu^2))\|\phi_\sigma\|^2+\tfrac 2\nu
  \sigma^2\inp{r}_{(\mathop{\mathrm{Re}}{}p^r)\phi_\sigma}.
\end{align*} Whence the contribution from $T_2$ to the bound
\eqref{eq:11.7.13.4.10cccd} is given by
\begin{align*}
  \cdots \leq C\epsilon^{-1}R_m \parb{(\sigma^{-1}+
    \tfrac\sigma{\nu^2})\|\phi_\sigma\|^2+\tfrac\sigma\nu\inp{r}_{(\mathop{\mathrm{Re}}{}p^r)\phi_\sigma}}+\epsilon
  R_m \sigma\|\tilde\phi_\sigma\|^2.
\end{align*} This bound agrees with \eqref{eq:11.7.13.4.10cccd} for
all large $\nu$ and $\sigma$ if we choose $\epsilon>0$ small (note
that $\nu >>R_m$ is used). Notice that we needed the second term on
the right hand side of \eqref{eq:11.7.13.4.10cccd} (this is the only
occurrence).

As for the contribution from $T_3$ we invoke
\eqref{eq:bnd_potentials}.

To treat the contribution from $T_{\rm com}$ we decompose
\begin{align*}
  T_{\rm com}=\chi^-_{\sigma^2/4}\parb{ \tilde p_{a}^2}^2[\widetilde
  H_{a},q_{a,R_{m+1}}^2]\tilde \phi_\sigma +[\widetilde
  H_{a},\chi^-_{\sigma^2/4}\parb{ \tilde
    p_{a}^2}^2]q_{a,R_{m+1}}^2\tilde \phi_\sigma
\end{align*} and use the representation \eqref{eq:rep1} for both terms
to the right.

Noting the following generalization of \eqref{eq:bnd_potentials}
\begin{align}
  \|\chi_\nu\mathrm{e}^{\sigma r}|p \phi|\|^2 \le
  4E\|\chi_\nu\mathrm{e}^{\sigma r}\phi\|^2 +C\langle
  \sigma\rangle^2\|\chi_{\nu/2}\mathrm{e}^{\sigma r}\phi\|^2,
  \label{eq:11.7.22.9.5bb}
\end{align} it follows (for the first term) that
\begin{align*}
  \|[\widetilde H_{a},q_{a,R_{m+1}}^2]\tilde \phi_\sigma\|\leq C
  \tfrac{\sigma}{R_{m+1}}\parb{\|\tilde \phi_\sigma\|+\| \phi\|},
\end{align*} which agrees with \eqref{eq:11.7.13.4.10cccd} provided
$R_{m+1}>> R_m$.

We claim that
\begin{align}\label{eq:final}
  \|[\widetilde H_{a},\chi^-_{\sigma^2/4}\parb{ \tilde
    p_{a}^2}^2]q_{a,R_{m+1}}^2\tilde \phi_\sigma \| \leq
  C\tfrac{\sigma}{\nu}\parb{\|\tilde \phi_\sigma\|+\| \phi\|},
\end{align} which also agrees with \eqref{eq:11.7.13.4.10cccd}, hence
finally showing the latter bound.

Now for showing \eqref{eq:final} we do various commutation using
\eqref{eq:11.7.22.9.5bb}
\begin{align*}
  &[\widetilde H_{a},\chi^-_{\sigma^2/4}\parb{ \tilde
    p_{a}^2}^2]q_{a,R_{m+1}}^2\tilde
  \phi_\sigma\\
  &=O\parb{\tfrac{\sigma}{\nu}}+[s(x)^{-2},\chi^-_{\sigma^2/4}\parb{ \tilde p_{a}^2}^2]H^{a}q_{a,R_{m+1}}^2\tilde \phi_\sigma\\
  &=O\parb{\tfrac{\sigma}{\nu}}+O\parb{\tfrac1{\sigma\nu}}\inp{
    \sigma^{-2}\tilde p_{a}^2}^{-1}s(x)^{-1}H^{a}q_{a,R_{m+1}}^2\phi_\sigma\\
  &=O\parb{\tfrac{\sigma}{\nu}}+O\parb{\tfrac{1}{\nu
      R_{m+1}}}+O\parb{\tfrac1{\sigma\nu}}\inp{ \sigma^{-2}\tilde
    p_{a}^2}^{-1}s(x)^{-1}q_{a,R_{m+1}}^2H^{a}\phi_\sigma,
\end{align*} where we used the convention
$\|O\parb{\tfrac{\sigma}{\nu}}\|\leq
C\tfrac{\sigma}{\nu}\parb{\|\tilde \phi_\sigma\|+\|\phi\|}$ and
similarly for the term $O\parb{\tfrac{1}{\nu R_{m+1}}}$.  Next we
decompose
\begin{align*}
  H^{a}\phi_\sigma&=\parb{(H^\sigma-E+\mathrm{i}\sigma
    (\mathop{\mathrm{Re}}{}p^r))\phi_\sigma+\mathrm{i}\mathrm{e}^{\sigma
      (r-4\nu)} R(\nu)\phi}\\&+(E-\mathrm{i}\sigma
  (\mathop{\mathrm{Re}}{}p^r)+\tfrac{\sigma^2}{2}|\mathrm{d}
  r|^2-\tfrac12 p_{a}^2-I_{a})\phi_\sigma-\mathrm{i}\mathrm{e}^{\sigma
    (r-4\nu)} R(\nu)\phi.
\end{align*} The first term vanishes. The second term contributes with
a term of the form $O\parb{\tfrac{\sigma}{\nu}}$. To see this we use
\eqref{eq:11.7.22.9.5bb} in two applications (note that the factor
$\inp{ \sigma^{-2}\tilde p_{a}^2}^{-1}s(x)^{-1}$ is used to bound one
factor of $p_{a}^2$), and we use the factor $q_{a,R_{m+1}}^2$ (note
that $q_{a,R_{m+1}}^2I_{a}$ is bounded). The third term is
$O\parb{\tfrac{1}{\nu}}$.  We have shown \eqref{eq:final}.

\section{ High-energy hard-core  one-body resolvent bound}\label{sec:One-body high-energy hard-core  resolvent bound}
In Subsections \ref{subsec: Concrete Besov
  spaces}--\ref{sec:CaseOmegarmc} we verify Condition
\ref{cond:general-scheme} for $N=2$ under Conditions
\ref{cond:smooth2} and \ref{cond:smooth3}.  The proof will be based on
various results for abstract Besov spaces to be given in Subsection
\ref{subsec: Abstract Besov spaces} and on a variant of Mourre theory
somewhat related to \cite{Sk4}.  We present our main results for the
obstacle case in Subsection \ref{subsec: Setting of problem}.  These
will be given in a slightly more general setting, and we devote 
Subsections \ref{subsec: Concrete Besov spaces} and \ref{subsec:
  Concrete Besov spaces,Proposition} to proofs.  The case of an empty
obstacle is treated in Subsection \ref{sec:CaseOmegarmc}.

\subsection{Abstract Besov spaces}\label{subsec: Abstract Besov
  spaces}

Let $A$ be a self-adjoint operator on a Hilbert space~$\vH$. Let
$R_0=0$ and $R_j=2^{j-1}$ for $j\in \N$.  We define correspondingly
characteristic functions $F_j=F(R_{j-1}\leq |\cdot|<R_j)$ and the
space
\begin{equation}
  \label{eq:1n}
  B=B(A)=\big \{u\in \vH \big | \,\sum_{j\in\N}R_j^{1/2}\|F_j(A)u\|=:\|u\|_B<\infty\big \}.
\end{equation}  We can identify (using the embeddings $\inp{A}^{-1}\vH\subset
B\subset \vH \subset B^*$, $\inp{A}:=\sqrt{A^2+1}$\,)
the dual  space $B^*$ as 
\begin{equation}
  \label{eq:2}
  B^*=B(A)^*=\big \{u\in \inp{A}\vH \big | \,\sup_{j\geq 1} R_j^{-1/2}\|F_j(A)u\|=:\|u\|_{B^*}<\infty\big \}.
\end{equation} Alternatively, the elements  $u$  of $B^*$ are  those
sequences   $u=(u_j)\subset  \vH $ with $u_j\in \Ran
(F_j(A))$ and $\sup_{j\in\N} R_j^{-1/2}\|u_j\|<\infty$.  For previous related works we refer to
\cite{AH,JP,GY,Wa,Ro,Sk4} 
and \cite[Subsections 14.1 and 30.2]{Ho2}. We note the
bounds, cf. \cite[Subsections 14.1]{Ho2},
\begin{equation}
  \label{eq:2p}
  \|u\|_{B^*}\leq \sup_{R>1} R^{-1/2}\|F(|A|<R)u\|\leq 2\|u\|_{B^*}. 
\end{equation}

Introducing \emph{abstract weighted spaces}
$L^2_s=L^2_s(A)=\inp{A}^{-s}\vH $ we have the embeddings
\begin{equation}\label{eq:3}
  L^2_s\subset B\subset L^2_{1/2}\subset \vH \subset
  L^2_{-1/2} \subset B^*\subset L^2_{-s},\mforall s>1/2.
\end{equation} All embeddings are continuous and corresponding
bounding constants can be chosen as absolute constants,
i.e. independently of $A$  and $\vH$. In particular
\begin{equation}
  \label{eq:22}
  \|u\|_{\vH}\leq \|u\|_B \mforall u\in B.
\end{equation}

We refer to the spaces $B$ and $B^*$ as \emph{abstract Besov spaces}.
Recall the following interpolation type result, here stated
abstractly. The proof is the same as that of the concrete versions
\cite[Theorem~2.5]{AH}, \cite[Theorem~14.1.4]{Ho2},
\cite[Proposition~2.3]{JP} and \cite[Subsection~4.3]{Ro}.
\begin{lemma}
  \label{lemma:resolvent-bounds} Let $A_1$ and $A_2$ be self-adjoint
  operators on Hilbert spaces $\vH_1$ and $\vH_2$, respectively, and
  let $s>1/2$. Suppose $T\in \vB(\vH_1,\vH_2)\cap
  \vB(L^2_s(A_1),L^2_s(A_2))$. Then $T\in \vB(B(A_1),B(A_2))$, and
  there is a constant $C=C(s)>0$ (independent of $T$) such that
  \begin{equation}
    \label{eq:4}
    \|T\|_{\vB(B(A_1),B(A_2))}\leq
    C\parb{\|T\|_{\vB(\vH_1,\vH_2)}+\|T\|_{\vB(L^2_s(A_1),L^2_s(A_2))}}.
  \end{equation}
\end{lemma}
We state and prove the following (partial) version of \cite[Lemma
2.5]{Sk4}.
\begin{lemma}
  \label{lemma:resolvent-boundsook} Suppose $A$ is a self-adjoint
  operator on a Hilbert space $\vH$, $c>1$ and $u\in B(A)$, then $u\in
  B(cA)$ with \begin{equation}
    \label{eq:4ooB}
    \|u\|_{B(cA)}\leq  8c^{1/2} \|u\|_{B(A)}.
  \end{equation}
\end{lemma}
\begin{proof} Pick $i\geq 2$ such that $R_{i-1}<c\leq R_i$. Then for
  all $j\geq i+1$
  \begin{equation*}
    F_j(ct)\leq F(R_{j-1}/R_i\leq |t|< R_{j}/R_{i-1})\leq F_{j-i+1}(t)+F_{j-i+2}(t).
  \end{equation*}
  Whence for any $u\in B(A)$ we can estimate
  \begin{align*}
    \|u\|_{B(cA)}&\leq \parb{ \sup_{j\geq
        i+1}\parb{R_{j}/R_{j-i+1}}^{1/2}+\sup_{j\geq
        i+1}\parb{R_{j}/R_{j-i+2}}^{1/2}}\|u\|_{B(A)}+\sum^{i}_{j=1}R_j^{1/2}\|u\|_{\vH}\\
    &\leq \parb{
      2^{(i-1)/2}+2^{(i-2)/2}+2^{i/2}(\sqrt 2 +1)}\|u\|_{B(A)}\\
    &\leq \parb{\sqrt 2+1+2(\sqrt 2 +1)}c^{1/2}
    \|u\|_{B(A)}\\
    &\leq 8c^{1/2} \|u\|_{B(A)}.
  \end{align*}
\end{proof}

We note the following abstract version of a result from \cite{JP, Mo2}
(proven by using suitable decompositions of unity and the
Cauchy Schwarz inequality, see also \cite[Subsection 2.2]{Wa}).
\begin{lemma}
  \label{lemma:resolvent-boundsooA} Let $A_1$ and $A_2$ be
  self-adjoint operators on Hilbert spaces $\vH_1$ and $\vH_2$,
  respectively, and let $T\in \vB(\vH_1,\vH_2)$. Suppose that
  uniformly in $m,n\in \Z$,
  \begin{equation}
    \label{eq:15hA}
    \|F(m\leq A_2< m+1)TF(n\leq A_1< n+1)\|\leq C.
  \end{equation} Then with the
  constant $C$ from \eqref{eq:15hA} we have 
  \begin{equation}
    \label{eq:21}
    \|T\|_{\vB(B(A_1),B(A_2)^*)}\leq  2C.
  \end{equation}
\end{lemma}
We note the following (partial) abstract criterion for
\eqref{eq:15hA}, cf.  \cite[(I.10)]{Mo2} (see also \cite{Wa}). Recall
that a bounded operator $T$ on a Hilbert space is called
\emph{accretive} if $T+T^*\geq 0$, cf.  for example \cite[Chapter
X]{RS}.
\begin{lemma}
  \label{lemma:resolvent-boundsooAb} Let $A$ be a self-adjoint operator
  on a Hilbert space  $\vH$, and suppose $T\in \vB(\vH)$ is accretive.
  Suppose the following bounds uniformly in $n\in \Z$,
  \begin{subequations}
    \begin{align*}
      \|F(n\leq A< n+1)TF(n\leq A< n+1)\|&\leq C_1,\\
      \|F(A<
      n)TF(n\leq A< n+1)\|&\leq C_2,\\
      \|F(n\leq A< n+1)TF(A\geq n)\|&\leq C_3.
    \end{align*}
  \end{subequations} Then \eqref{eq:15hA} holds with $A_1=A_2=A$, the
  accretive $T$ and with $C=2C_1+C_2+C_3$.
\end{lemma}

\subsection{Setting of problem}\label{subsec: Setting of problem}
We assume that $\Omega\subset\bX$ is open and that
$\Theta:=\bX\setminus\overline {\Omega}\neq\emptyset$ is bounded with
smooth boundary $\partial \Theta=\partial \Omega$. Moreover we assume
that $\Theta$ is strictly convex, see Appendix \ref{sec:Appendix} for definition.
The case $\Omega=\bX$ is simpler and will be treated in Subsection
\ref{sec:CaseOmegarmc}.

We consider a Hilbert space $\vH=L^2(\Omega,\d x)\otimes L^2(M,\d
y)$. The structure of the second factor will not be of importance. To
make contact to \eqref{eq:Basic_assump} we think of $x$ as $x^{b}$ and
$y$ as $x_{b}$ (here $b\notin\{a_{\min}, a_{\max}\}$ and $\#
b=N$). Hence the function $s$ of \eqref{eq:Basic_assump} is now a
function of $x$ and $y$, viz. $s=s(x,y)$. The operator $\widetilde
H_{b}$ takes the form $\widetilde H_{b}=\tilde H^{b}+\tilde B_b$ on
$\vH$. We simplify notation and look at
\begin{align*}
  H&=\widetilde H^{b}+\widetilde B;\\
  \widetilde H^{b}&=s(x,y)^{-1}(\tfrac12p_x^2+V(x))s(x,y)^{-1}=\tilde
  p_x^2+V(x)s(x,y)^{-2},\\
  \widetilde B&=\widetilde B(x),
\end{align*} where the operator $\widetilde B(x)$ acts as a bounded
operator on the component $L^2(M,\d y)$. As an operator on $\vH$ it is
bounded, and it needs to be small and regular in $x$ in a certain
sense (to be specified in
\eqref{eq:smallness2}-\eqref{eq:smallness3}). Whence our method does
not require much specific structure of the operator-valued potential
$\widetilde B$.  The unbounded part, $\widetilde H^{b}$, is defined
with Dirichlet boundary condition at $\partial \Omega$, and the
two-body potential $V=V(x)$ needs to be sufficiently
regular. For simplicity we impose $V\in C^\infty(\overline \Omega)$
and  $
      \partial ^\alpha_xV(x)=O\big(|x|^{-\varepsilon-|\alpha|}\big )$, 
cf. Condition \ref{cond:smooth2} \ref{item:cond12b}. Whence the form domain of
$\widetilde H^{b}$  is given by the space
\begin{align*}
  Q(\widetilde H^{b})=Q(H)=L^2(M,H^1_0(\Omega);\d y)\subset\vH.
\end{align*}
The reader should keep in mind the rough approximation $\widetilde
H^{b}\approx -\Delta_x+2V(x)$ (recall here that $s\approx 1/\sqrt 2$ in
the large $\nu$ regime). We denote the resolvent of
$H_\sigma:=\sigma^{-2}H$ by $R(z,\sigma)$,
viz. $R(z,\sigma)=(H_\sigma-z)^{-1}$.

\begin{subequations}
  We introduce a function $r=r(x)$ that is different from the function
  $r$ of Section \ref{sec:Reduction to high-energy hard-core
    sub-system resolvent bounds}. It is now given as
  \begin{align}\label{eq:r}
    r(x)=\dist (x,\partial \Omega),
  \end{align} which can be extended to a smooth function on $\bX$ and
  which at infinity has bounds
  \begin{align}
    \label{eq:r_bounds}
    \partial^\alpha r=O\parb{r^{1-|\alpha|}}=O\parb{\inp{x}^{1-|\alpha|}}.
  \end{align} More importantly there exists $c>0$ such that
  \begin{align}
    \label{eq:low_bound}
    \nabla ^2 r_{|\{r(x)=r\}}\geq \tfrac c{1+r}I.
  \end{align}
\end{subequations} The verification of \eqref{eq:r_bounds} and
\eqref{eq:low_bound} is given in Appendix
\ref{sec:Appendix}.
Note the following consequence of \eqref{eq:low_bound},
\begin{align*}
  \forall \delta\in (0,1]:\tfrac{\nabla^2 r^2}{2}\geq\delta \d
  r\otimes\d r\oplus r\nabla ^2 r_{|\{r(x)=r\}}\geq\min (\delta,c)\tfrac r{1+r}I.
\end{align*}

In terms of the function $r$ we introduce a conjugate operator
different from the operator $A$ that appears in Lemma
\ref{lemma:Mourre1}. Now
\begin{align}\label{eq:5A}
  A:=\tfrac{\nabla r^2}{2}\cdot p+p\cdot \tfrac{\nabla r^2}{2}.
\end{align} This operator is self-adjoint on $L^2(\Omega,\d x)$ (it is
essentially self-adjoint on $C^\infty_{\c}(\Omega)$) and whence also
on $\vH$. Note that $Q(H)$ is ``boundedly stable'' under the dynamics
generated by $A$ (using here terminology of \cite{GGM}, see also
\cite{FMS}), i.e.
\begin{subequations}
  \begin{align}
    \label{eq:banu_stable}
    \forall \psi\in Q(H):\,\sup_{|t|<1}\|\e^{\i
      tA}\psi\|_{Q(H)}<\infty.
  \end{align}
  We note the representation $A=rp_r+p_rr$ where
  \begin{align*}
    p_r:=\tfrac{\nabla r}{2}\cdot p+p\cdot \tfrac{\nabla r}{2}=-\i
    \tfrac{\partial}{\partial r}-\i\tfrac{\Delta r}{2}.
  \end{align*} In turn the  operator $p_r$ is symmetric as an operator with
  domain $H^1_0(\Omega)$, and we define $p_r^2$ as the Friedrichs
  extension from $C^2_\c(\Omega)$ and use the same notation for
  $p_r^2\otimes I$. Note the inclusion $Q(H)\subset Q(p_r^2)$ for form
  domains as well as the following analogue of \eqref{eq:banu_stable}
  \begin{align}
    \label{eq:banu_stableb}
    \forall \psi\in Q(p_r^2):\,\sup_{|t|<1}\|\e^{\i
      tA}\psi\|_{Q(p_r^2)}<\infty.
  \end{align}
\end{subequations} Note for \eqref{eq:banu_stable} and
\eqref{eq:banu_stableb} that a similar  property is derived in
Appendix \ref{sec:Appendixb} for the conjugate operator used in Section \ref{sec:Reduction to high-energy hard-core sub-system
    resolvent bounds} (the one constructed by  the Graf
vector field). Note for  \eqref{eq:banu_stableb} the explicit formula  $\|p_r\e^{\i
      tA}\psi\|=\e^{2t}\|p_r\psi\|$. For a  different proof, given in a
    generalized setting, see Lemma \ref{lem:12.6.4.13.14b}. The
    property  \eqref{eq:banu_stable} follows from Lemma \ref{lem:12.6.4.13.14}.
  
We recall the Hardy bounds, cf. \cite[Lemma 5.3.1]{Da},
\begin{align}
  \label{eq:hardy}\|r^{-\kappa}|p_r|^{-\kappa}\|\leq 2\text{ for }\kappa\in [0,1].
\end{align} Moreover we have
\begin{align}
  \label{eq:laplace}
  -\Delta_x=p_r^2+L^2+\tfrac14 (\Delta
    r)^2+\tfrac12 (\partial_r\Delta r),
\end{align} where the second term is positive and commutes with $r$
(it is the Laplace-Beltrami operator in geodesic coordinates), and the
third and fourth terms are  bounded functions on $\Omega$.

We also introduce operators $f_1,f_2\geq 0$ with squares
\begin{align*}
  f_1^2&=\sigma^{-2/3}+\tfrac r{1+r},\\
  f_2^2&=\sigma^{-2/3}+\tfrac r{1+r}
  +\sigma^{-2}p_r^2=f_1^2+\sigma^{-2}p_r^2.
\end{align*}
A main preliminary bound of this section is
\begin{lemma}
  \label{lemma:resolvent-boundsA} With $A$ given by \eqref{eq:5A} we
  have uniformly in all large $\sigma,\nu>1$ and all $\Re z\approx 1$
  \begin{equation}
    \label{eq:limitbound3a}
    \left\|
      f_{2} R(z,\sigma)
      f_{2} \right\|_{\vB (B(A),B(A)^*)} \leq C.
  \end{equation}
\end{lemma}

The main result of the section is
\begin{proposition}\label{prop:setting-problem}
  With $r$ given as the multiplication operator on $\vH$ in terms of
  the function \eqref{eq:r} we have uniformly in all large
  $\sigma,\nu>1$ and all $\Re z\approx 1$
  \begin{equation}
    \label{eq:limitbound3a_main}
    \left\|
      R(z,\sigma)
    \right\|_{\vB (B(r),B(r)^*)} \leq C\sigma.
  \end{equation}
\end{proposition} Obviously for  $b$ with $\#b=N$ and  $\Omega_b\neq \bX^b$, and under the regularity
conditions on $\Omega=\Omega_b$ and $V_b=V$ introduced above, the bound
\eqref{eq:Basic_assump} is a
consequence of Proposition \ref{prop:setting-problem}.

\subsection{Besov space bound of resolvent, Lemma
  \ref{lemma:resolvent-boundsA}}\label{subsec: Concrete Besov spaces}

In this subsection we shall prove Lemma \ref{lemma:resolvent-boundsA}
using a variant of Mourre theory.

\subsubsection{First order commutator}
We ``compute'' the commutator
\begin{align}
  \begin{split}
    \label{eq:1_comm}
    \i[H,A]&:=s^{-1}\parb{2p_r^2 +2
      p_ir(\nabla^2r)^{ij}p_j+W}s^{-1}\\&\;\;\;\;+2\Re\parb{s^{-1}(\nabla
      r^2\cdot \nabla_xs)\widetilde H^b}-\nabla r^2\cdot
    \nabla_x\widetilde B;
  \end{split}\\
  W(x)&:=\tfrac12 (\Delta r)^2+\partial_r\Delta
  r-\tfrac14\Delta^2 r^2-\nabla r^2\cdot \nabla V(x).\nonumber
\end{align} Thus at this stage the first order commutator $\i[H,A]$ is
defined by its formal expression. We note that it is a bounded
quadratic form on $Q(H)$. The term $W$ is a bounded function, and the
second and third terms are ``small''. More precisely in terms of the
parameters $\nu$ and $\sigma$ of Section \ref{sec:Reduction to
  high-energy hard-core sub-system resolvent bounds} we have uniform
bounds, cf. \eqref{eq:small_termb}, \eqref{eq:der_s} and
\eqref{eq:rep1},
\begin{subequations}
  \begin{align}
    \label{eq:smallness1}
    |\tfrac{1+r}{rs}\nabla
    r^2\cdot \nabla_xs|&\leq C{\nu}^{-1},\\
    \label{eq:smallness1b}
    |\tfrac{1+r}{r}(\nabla
    r^2\cdot \nabla_x)^2s|&\leq C,\\
    \label{eq:smallness2}
    \|\tfrac{1+r}{r}\nabla r^2\cdot
    \nabla_x\widetilde B\|&\leq C \tfrac{\sigma^2}{\nu},\\
    \label{eq:smallness2b}
    \|\tfrac{1+r}{r}(\nabla r^2\cdot
    \nabla_x)^2\widetilde B\|&\leq C \sigma^2,\\
    0\leq \widetilde B&\leq \tfrac78\sigma^2. \label{eq:smallness3}
  \end{align}
\end{subequations}

We can estimate the second term after commutation as
\begin{align*}
  2\Re\parb{s^{-1}(\nabla r^2\cdot \nabla_xs)\widetilde H^b}\geq -C_1
  \nu^{-1}\parb{\Re\parb{\tfrac r{1+r}H}+C_2},
\end{align*} cf. \eqref{eq:smallness1} and \eqref{eq:smallness3}.

Similarly we can estimate the third term as
\begin{align*}
  -\nabla r^2\cdot \nabla_x\widetilde B\geq -C\tfrac {\sigma^2}\nu
  \tfrac r{1+r},
\end{align*} cf. \eqref{eq:smallness2}.

Using these bounds, \eqref{eq:low_bound} and \eqref{eq:hardy} we can
estimate for some small $\delta>0$ (and uniformly in $\sigma,\nu>1$)
\begin{align}
  \begin{split}
    \label{eq:1_commbb}
    \i[\sigma^{-2}H,A]&\geq \sigma^{-2}3p_r^2 +\delta
    \parbb{\sigma^{-2}r^{-2} +\Re\parb{\tfrac
        r{1+r}\sigma^{-2}\widetilde H^b}}\\&-C_1
    \nu^{-1}\Re\parb{\tfrac
      r{1+r}\sigma^{-2}H}-C_3(\sigma^{-2}+\nu^{-1}\tfrac r{1+r}).
  \end{split}
\end{align}

Next we estimate using \eqref{eq:smallness3}
\begin{align*}
  \Re\parb{\tfrac r{1+r}\sigma^{-2}\widetilde H^b}\geq(\Re
  z-\tfrac78)\tfrac r{1+r}+\Re\parb{\tfrac r{1+r}(\sigma^{-2} H-z)}.
\end{align*}

We are interested in the regime $\Re z\approx 1$. Concretely let us
assume that $|1-\Re z |\leq \tfrac19$ allowing us to estimate
uniformly in the spectral parameter: There exists $\delta'\in(0,3)$
such that for all such $z$ and  all large $\sigma,\nu>1$
\begin{align*}
  \delta
  \parbb{\sigma^{-2}r^{-2} +(\Re z-\tfrac78)\tfrac r{1+r}}-C_1
  \nu^{-1}\Re z \tfrac r{1+r}-C_3(\sigma^{-2}+\nu^{-1}\tfrac
  r{1+r})\geq \delta'f_1^2.
\end{align*} From \eqref{eq:1_commbb} we thus obtain
\begin{subequations}
  \begin{align}\label{eq:posCom1}
    \i[H_\sigma,A]\geq
    \delta'f_2^2+(\delta-C_1\nu^{-1})\Re\parb{\tfrac r{1+r}
      (H_\sigma-z)}.
  \end{align}

  Now let us introduce (cf. the method of \cite{Mo1})
  \begin{align*}
    R_{z}(\epsilon) = (H_\sigma - \i\epsilon \i[H_\sigma, A] -
    z)^{-1};\, \epsilon \,\Im z>0,\,|1-\Re z |\leq \tfrac19.
  \end{align*} We only need $|\epsilon|\leq 1$, and we note that as a
  form $H_\sigma - \i\epsilon \i[H_\sigma, A] $ is strictly
  $m$-sectorial in the terminology of \cite{Ka,RS}, cf. the
  computation \eqref{eq:1_comm}. The associated operator,
  cf. \cite[Theorem VIII.17]{RS}, is invertible if we also assume that
  $|\Im z|>>1$, and hence the inverse is well-defined with  adjoint $
  R_{z}(\epsilon)^* = R_{\bar z}(-\epsilon) $ under these
  conditions. However it follows from a connectedness argument and
  \eqref{eq:quadratic} stated below (with $T=f_2$) that
  $R_{z}(\epsilon)$ is well-defined without the condition $|\Im
  z|>>1$. Note also that $\lim_{\epsilon \to
    0}R_{z}(\epsilon)=R(z,\sigma)$.

  We obtain from \eqref{eq:posCom1} that
  \begin{align}\label{eq:posCom2}
    \i[H_\sigma,A]\geq
    \tfrac{\delta'}2f_2^2+(\delta-C_1\nu^{-1})\Re\parb{\tfrac r{1+r}
      (H_\sigma- \i\epsilon \i[H_\sigma, A] -z)}.
  \end{align}
\end{subequations}
In fact using \eqref{eq:laplace} and \eqref{eq:1_comm} we compute
\begin{align}
  \begin{split}\label{eq:imag2com}
    &\Re\parb{\tfrac r{1+r} \i\epsilon \i[H_\sigma, A]}\\&=\tfrac
    \epsilon {2\sigma^2}\i[\tfrac r{1+r}, \i[H, A]]\\&=\tfrac \epsilon
    {2\sigma^2}\i[\tfrac r{1+r},
    s^{-1}2p_r^2s^{-1}+\Re\parb{s^{-1}(\nabla
      r^2\cdot \nabla_xs)s^{-1}p_r^2s^{-1}}]\\
    &=-\tfrac \epsilon {\sigma^2}s^{-1}\parbb{2\Re \parb{p_r
        (1+r)^{-2}}+\Re \parb{(\nabla
        r^2\cdot \nabla_xs)s^{-1}\Re \parb{p_r  (1+r)^{-2}}}}s^{-1}\\
    &=-\tfrac \epsilon {\sigma^2}\Re \parb{p_rg},\end{split}
\end{align} where $g=g(x,y)$ is a uniformly bounded function, and thus
indeed
\begin{align*}
  -(\delta-C_1\nu^{-1})\Re\parb{\tfrac r{1+r}\i\epsilon \i[H_\sigma,
    A]}\leq C\sigma^{-2/3}f_2^2\leq\tfrac{\delta'}2f_2^2.
\end{align*}

Due to \eqref{eq:posCom2} and the second resolvent equation we have the quadratic estimate
\begin{align*}
  \|f_2 R_{z}(\epsilon) T \|^2 \leq C_1 \parb{|\epsilon|^{-1} \|T^*
    R_{z}(\epsilon) T \|+\|T^* R_{z}(\epsilon)^* \tfrac r{1+r}T \|}.
\end{align*} Hence if  $T$ is an operator obeying 
\begin{subequations}
  \begin{align}
    \label{eq:comT}
    \|f_2^{-1}\tfrac r{1+r}T \|\leq C_2,
  \end{align} then
  \begin{align*}
    \|f_2 R_{z}(\epsilon) T \|^2 \leq C_1  \parb{|\epsilon|^{-1} \|T^*
    R_{z}(\epsilon) T \|+C_2\|T^* R_{z}(\epsilon)^* f_2 \|}.
  \end{align*} This leads to
  \begin{align}
    \label{eq:quadratic}
    \|f_2 R_{z}(\epsilon) T \|^2 \leq C_3 |\epsilon|^{-1} \|T^*
    R_{z}(\epsilon) T \|+C_4.
  \end{align}
\end{subequations}
We have the examples $T=f_2$ and $T=f_2\inp{A}^{-1}$ with bounds
independent of all large $\sigma$ and $\nu$. Indeed for all $\psi\in
\vH$ (or alternatively for all $\psi\in L^2(\Omega,\d x)$ since the
operators act on the first tensor factor only)
\begin{align*}
  \inp{f_2^2}_{\tfrac r{1+r}f^{-1}_2\psi}&\leq
  \|\psi\|^2+\inp{\sigma^{-2}p_r^2}_{\tfrac r{1+r}f^{-1}_2\psi} \\
  &\leq
  \|\psi\|^2+\|\sigma^{-1}p_rf^{-1}_2\psi\|^2+C_1\sigma^{-2}\|f^{-1}_2\psi\|^2\\
  &\leq C_2\|\psi\|^2,
\end{align*} proving \eqref{eq:comT} for these examples. 

\subsubsection{Technical lemma} For $B\in\vB(Q(p_r^2), \vH)$ we define
\begin{align*}
  \ad_A(B)= [B,A]=\slim_{t\to 0}\i t^{-1}\parb{B\e^{-\i
      tA}-\e^{-\i tA}B}_{|Q(p^2_r)\to \vH},
\end{align*} provided the right hand side exists. Note that we here 
  use \eqref{eq:banu_stableb}. In the terminology
of \cite{GGM},  $B\in C^1(A_{|Q(p^2_r)},A_{|\vH})$ if the right hand
side exists. We use this interpretation of
the (repeated) commutators in the following lemma (in turn to be used later).
\begin{lemma}\label{lem:first-order-comm} Uniformly in all
  $\sigma>1$
  \begin{subequations}
    \begin{align}
      \label{eq:Afcom}
      \|\ad_A(f_2)f_2^{-1}\|&\leq C,\\
      \|\ad^2_A(f_2)f_2^{-1}\|&\leq C.\label{eq:Afcomb}
    \end{align} 
\end{subequations}
\end{lemma}
\begin{proof}
  Note the representation (valid for any strictly positive operator
  $S$)
  \begin{equation*}
    S^{-1/2} = \pi^{-1}\int_0^\infty s^{-1/2} (S+s)^{-1}\d s.
  \end{equation*} With $S=f_2^2$ we thus obtain (for the first term)
  \begin{align*}
    [f_2,A]&=S^{-1/2}[S,A]+[S^{-1/2},A]S\\&=S^{-1/2}[S,A]-\pi^{-1}\int_0^\infty
    s^{-1/2} (S+s)^{-1}[S,A](S+s)^{-1}\d sS,
  \end{align*} where 
\begin{align*}
  [S,A]=\slim_{t\to 0}\i t^{-1}\parb{S\e^{-\i
      tA}-\e^{-\i tA}S}_{|Q(p^2_r)\to Q(p^2_r)^*}
\end{align*} is computed (up a factor $-\i$) as 
  \begin{align*}
    \i[S,A]=2\Re(\i[S, rp_r])=4\sigma^{-2}p_r^2-2r
    (1+r)^{-2}=4S-4f_1^2-2r (1+r)^{-2}.
  \end{align*} In particular
  \begin{align*}
    - CS\leq \i [S,A]\leq C S,
  \end{align*} which obviously allows  us to conclude that
  \begin{align*}
    B_1:=S^{-1/2}[S,A]S^{-1/2}
  \end{align*} is bounded. Similarly we introduce
  \begin{align*}
    B_2:=\int_0^\infty s^{-1/2} (S+s)^{-1}[S,A](S+s)^{-1}\d sS^{1/2},
  \end{align*} and it remains to show boundedness of $B_2$: We write
  with $f:=(f_1^2+\tfrac12 r (1+r)^{-2})^{1/2}$
  \begin{align*}
    B_2&= C I+\i\int_0^\infty s^{-1/2} (S+s)^{-1}4f^2(S+s)^{-1}\d sS^{1/2}\\
    &= C (I-f^2S^{-1})+\i\int_0^\infty s^{-1/2}
    [(S+s)^{-1},4f^2](S+s)^{-1}\d sS^{1/2}\\
    &= C (I-fS^{-1}f-f[f,S^{-1}])-4\i\int_0^\infty s^{-1/2}
    (S+s)^{-1}[S,f^2](S+s)^{-2}\d sS^{1/2}\\
    &= B-CfS^{-1}[S,f]S^{-1}-4\i\int_0^\infty s^{-1/2}
    (S+s)^{-1}[S,f^2](S+s)^{-2}\d sS^{1/2}.
  \end{align*} Using the notation $O(1)=O_{\vB(\vH)}(\sigma^0)$ we
  have
  \begin{align*}
    \i [S,f]&=\tfrac {p_r}\sigma O(1)(\sigma f)^{-1}+\text{ h.c.},\\
    \i [S,f^2]&=\tfrac {p_r}\sigma O(1)\sigma^{-1}+\text{
      h.c.}=\sigma^{-2/3}S^{1/2}O(1)S^{1/2},
  \end{align*} and the first identity yields that also
  \begin{align*}
    \i fS^{-1}[S,f]S^{-1}&=\parb{fS^{-1}\tfrac {p_r}\sigma}
    O(1)(\sigma
    f)^{-1}S^{-1}+fS^{-1}(\sigma^{2/3} f)^{-1}O(1)\parb{\sigma^{-1/3}\tfrac {p_r}\sigma S^{-1}}\\
    &=O(1)\sigma^{-2/3}S^{-1}+fS^{-1}\sigma^{-1/3}O(1)=O(1);
  \end{align*} i.e. the term is uniformly bounded. The second identity
  yields that the integral is bounded by
  \begin{align}\label{eq:estcoms}
    \begin{split}
      &C\sigma^{-2/3}
      \int_0^\infty s^{-1/2} \|(S+s)^{-1}S^{1/2}(S+s)^{-1}\|\d s\\
      &\leq C\sigma^{-2/3} \int_0^\infty s^{-1/2}
      (\sigma^{-2/3}+s)^{-3/2}\d s,
    \end{split}
  \end{align} and since the latter integral is independent of $\sigma$
  indeed also the integral is uniformly bounded. So also $B_2$ is
  uniformly bounded and \eqref{eq:Afcom} follows.

  As for \eqref{eq:Afcomb} we use a previous computation to obtain
  \begin{align*}
    -\ad^2_A(S)&=4\parb{4S-4f_1^2-2r
      (1+r)^{-2}}+2r\tfrac{\partial}{\partial
      r}\parb{4f_1^2+2r(1+r)^{-2}}\\
    &=16S-16\tilde f^2;\; \tilde f=\parbb{\sigma^{-2/3}+\frac{\tfrac34
        r +\tfrac94 r^2+r^3}{(1+r)^3}}^{1/2}.
  \end{align*} which leads to form-boundedness
  \begin{align*}
    \|S^{-1/2}\ad^2_A(S)S^{-1/2}\|\leq C.
  \end{align*}

  We decompose
  \begin{align*}
    \ad^2_A(f_2)S^{-1/2}&=T_1+\cdots+T_5;\\
    T_1&=\ad_A(S^{-1/2})\,\ad_A(S)S^{-1/2}=\parb{\ad_A(S^{-1/2})S^{1/2}}\parb{S^{-1/2}
      \ad_A(S)S^{-1/2}},\\
    T_2&=S^{-1/2}\ad^2_A(S)S^{-1/2},\\
    T_3&=-\pi^{-1}\int_0^\infty s^{-1/2}
    (S+s)^{-1}\ad^2_A(S)(S+s)^{-1}\d
    sS^{1/2},\\
    T_4 &=\pi^{-1}\int_0^\infty s^{-1/2}
    (S+s)^{-1}\ad_A(S)(S+s)^{-1}\ad_A(S)(S+s)^{-1}\d sS^{1/2},\\
    T_5 &=-\pi^{-1}\int_0^\infty s^{-1/2}
    (S+s)^{-1}\ad_A(S)(S+s)^{-1}\ad_A(S)(S+s)^{-1}s\d sS^{-1/2}.
  \end{align*}
  The boundedness of the term $T_1$ follows from the previous
  proof. Clearly the term $T_2$ is bounded.  We can show boundedness
  of $T_3$ as we proceeded for \eqref{eq:Afcom} (note that now $\tilde
  f$ plays the role of the previous $f$).  For $T_4$ and $T_5$ we
  rewrite
  \begin{align*}
    T_4+T_5&=\widetilde T_4+\widetilde T_5;\\
    \widetilde T_4&=\pi^{-1}\int_0^\infty s^{-1/2}
    (S+s)^{-1}\ad_A(S)(S+s)^{-1}\d s\,S^{1/2}\parb{S^{-1/2}\ad_A(S)S^{-1/2}},\\
    \widetilde T_5 &=-2\pi^{-1}\int_0^\infty s^{-1/2}
    (S+s)^{-1}\ad_A(S)(S+s)^{-1}\ad_A(S)(S+s)^{-1}s\d sS^{-1/2}.
  \end{align*} The boundedness of the term $\widetilde T_4$ follows
  from the previous proof. Whence it only remains to show boundedness
  of $\widetilde T_5$.  We proceed in a similar fashion as before
  substituting for the first factor of $\ad_A(S)$ from the left
  $\ad_A(S)=-\i4(S-f^2)$ and then move to the left. The commutator is
  treated as in \eqref{eq:estcoms} using now also the form-boundedness
  of the second factor of $\ad_A(S)$. So it remains to consider
  \begin{align*}
    (I-f^2S^{-1})\int_0^\infty s^{-1/2}
    S(S+s)^{-2}\ad_A(S)(S+s)^{-1}S^{-1/2}s\d s.
  \end{align*} We saw before that the first factor is bounded. For the
  integral we substitute again $\ad_A(S)=-\i4(S-f^2)$ and move to the
  left. Estimating as in \eqref{eq:estcoms} we conclude that the
  commutator is bounded. So we are left with
  \begin{align*}
    \int_0^\infty \cdots \d s=C_1(S-f^2)\int_0^\infty s^{-1/2}
    S(S+s)^{-3}S^{-1/2}s\d s=C_2(I-f^2S^{-1}),
  \end{align*} which is bounded. Whence \eqref{eq:Afcomb} follows.
\end{proof}

\subsubsection{Second order commutator}
In \eqref{eq:1_comm} we took the formal commutator as a definition of
$\i[H,A]$, however due to the property \eqref{eq:banu_stable} there is
the following alternative interpretation
\begin{align*}
  -\i [H,A]=\slim_{t\to 0}t^{-1}\parb{H\e^{-\i tA}-\e^{-\i
      tA}H}_{|Q(H)\to Q(H)^*},
\end{align*} cf. Appendix A and \cite{GGM}, which allows us to compute
\begin{equation}
  \label{eq:7}
  \frac{\d}{\d\epsilon} R_{z}(\epsilon) =-R_{z}(\epsilon)[H_\sigma,A]R_{z}(\epsilon)=
  R_{z}(\epsilon)A -AR_{z}(\epsilon)  +\epsilon R_{z}(\epsilon)\ad^2_A(H_\sigma)R_{z}(\epsilon),
\end{equation} where $\ad^2_A(H_\sigma)=[[H_\sigma,A],A]\in
\vB(Q(H),Q(H)^*)$. Note that in the terminology
of \cite{GGM},  $H\in C^2(A_{|Q(H)},A_{| Q(H)^*})$. The second identity of \eqref{eq:7} is  valid as a form on
the domain $\vD^*:=\vD\parb{A_{| Q(H)^*}}$ of the generator of the extended group
$\{\e^{-\i tA}\}_{|Q(H)^*}$, so that  indeed $A:\vD^*\to Q(H)^*$,
which combines with the mapping property 
$R_{z}(\epsilon):Q(H)^*\to Q(H)$. Below we use tacitly this
interpretation and the fact that $f_{2}\inp{A}^{-1}:\vH\to \vD^*$,
cf. \eqref{eq:Afcom}.

Using \eqref{eq:hardy} (with $\kappa=1/2$), \eqref{eq:1_comm} and
\eqref{eq:smallness1}-\eqref{eq:smallness2b} we compute
\begin{align}\label{eq:secCom}
  \ad^2_A(H_\sigma)=f_2B_0f_2+\sum^{\dim \bX}_{i,j=1}\parb{\tfrac
    r{1+r}}^{1/2}(\sigma s)^{-1}p_iB_{ij}p_j(\sigma s)^{-1}\parb{\tfrac
    r{1+r}}^{1/2},
\end{align} where $p_j$ denotes the components of $p_x$ and all $B$'s
are uniformly bounded.

Using \eqref{eq:quadratic}, \eqref{eq:7} and \eqref{eq:secCom} we
shall prove three bounds which are uniform in $ z$ and $\epsilon$ as
specified above and (for convenience) with $\Im z, \epsilon>0$  as
well as uniform in (large) $\sigma$ and $\nu$:
\begin{subequations}
  \begin{align}
    \label{eq:8}
    \|F_z(\epsilon)\|&\leq C \mfor
    F_z(\epsilon):=\inp{A}^{-1}f_{2}R_{z}(\epsilon)f_{2}\inp{A}^{-1},\\
    \label{eq:9}\|F^-_z(\epsilon)\|&\leq C \mfor
    F^-_z(\epsilon):=\e^{\epsilon A}F(A<
    0)f_{2}R_{z}(\epsilon)f_{2}\inp{A}^{-2},\\
    \label{eq:10}\|F^+_z(\epsilon)\|&\leq C \mfor
    F^+_z(\epsilon):=\inp{A}^{-2}f_{2}R_{z}(\epsilon)f_{2}F(A\geq
    0)\e^{-\epsilon A}.
  \end{align}
\end{subequations}

\noindent{\bf Re \eqref{eq:8}.} Due to \eqref{eq:quadratic} for
$T=f_2\inp{A}^{-1}$ (note that we proved \eqref{eq:comT})
\begin{equation}
  \label{eq:11n}
  \|F_z(\epsilon)\|\leq C\epsilon^{-1}\mfor 0<\epsilon\leq 1.
\end{equation}   Obviously \eqref{eq:Afcom}  yields  the bounds
\begin{equation}
  \label{eq:12n}
  \|f_{2}^{-1}Af_{2}\inp{A}^{-1}\|\leq C\mand \|\inp{A}^{-1}f_{2}Af_{2}^{-1}\|\leq C.
\end{equation} Exploiting  \eqref{eq:quadratic}, \eqref{eq:7}, \eqref{eq:secCom}  and
\eqref{eq:12n} we can show that 
\begin{equation}
  \label{eq:13n}
  \big \|\frac{\d}{\d\epsilon} F_{z}(\epsilon) \big \|\leq C\parb{
    \epsilon^{-1/2}\|F_z(\epsilon)\|^{1/2}+\|F_z(\epsilon)\|+\widetilde
  C}.
\end{equation} Here we can argue as follows for the contribution to \eqref{eq:7}
from the second term in \eqref{eq:secCom}.
For $\psi\in \vH $  we estimate using \eqref{eq:hardy},
\eqref{eq:smallness3}  and \eqref{eq:imag2com}
\begin{subequations}
\begin{align}\label{eq:dobbelt}
  \begin{split}
    \sum_j&\|p_j(\sigma s)^{-1}\parb{\tfrac r{1+r}}^{1/2}\psi\|^2\leq
    C_1\|f_2\psi\|^2+2\Re \inp{H_\sigma-z}_{
      r^{1/2}(1+r)^{-1/2}\psi}\\
    &\leq C_2\|f_2\psi\|^2+2\Re \inp{\tfrac r{1+r}(H_\sigma-
      \i\epsilon \i[H_\sigma, A] -z }_{\psi}.
  \end{split}
\end{align} We use \eqref{eq:dobbelt} to
$\psi=R_{z}(\epsilon)f_{2}\inp{A}^{-1}\tilde \psi$, $\tilde \psi\in \vH $,   and  then \eqref{eq:comT}
and \eqref{eq:quadratic} with $T=f_2\inp{A}^{-1}$. Similarly we apply
\begin{align}\label{eq:dobbelt2}
  \sum_i\|p_i(\sigma s)^{-1}\parb{\tfrac r{1+r}}^{1/2}\psi\|^2\leq C_2\|f_2\psi\|^2+2\Re \inp{\tfrac r{1+r}(H_\sigma+
      \i\epsilon \i[H_\sigma, A] -\bar z }_{\psi}
\end{align} 
  \end{subequations}
to  
$\psi=R_{\bar z}(-\epsilon)f_{2}\inp{A}^{-1}\tilde \psi$. We conclude  \eqref{eq:13n}.

Clearly \eqref{eq:8} follows from \eqref{eq:11n} and \eqref{eq:13n} by
two integrations.

\noindent{\bf Re \eqref{eq:9}.} Due to \eqref{eq:quadratic} and
\eqref{eq:8}
\begin{equation}
  \label{eq:11-}
  \|F^-_z(\epsilon)\|\leq C\epsilon^{-1/2}.
\end{equation} Using \eqref{eq:7} we compute 
\begin{align}
  \label{eq:14}
  \frac{\d}{\d\epsilon} F^-_{z}(\epsilon)&=T_1+T_2+T_3;\\
  T_1&= \e^{\epsilon A}F(A<
  0)[A,f_{2}]R_{z}(\epsilon)f_{2}\inp{A}^{-2},\nonumber\\
  T_2&= \e^{\epsilon A}F(A<
  0)f_{2}R_{z}(\epsilon)Af_{2}\inp{A}^{-2},\nonumber\\
  T_3&= \epsilon\e^{\epsilon A}F(A<
  0)f_{2}R_{z}(\epsilon)\ad^2_A(H_\sigma)R_{z}(\epsilon)f_{2}\inp{A}^{-2}.\nonumber
\end{align} Using again \eqref{eq:quadratic} and \eqref{eq:8} we can
estimate
\begin{equation}
  \label{eq:11-b}
  \|T_j\|\leq C\epsilon^{-1/2}\mfor 0<\epsilon\leq 1\mand j=1,2,3.
\end{equation}  Notice that for all of the terms $T_1$--$T_3$ we  apply \eqref{eq:quadratic}
with $T=f_{2}\inp{A}^{-1}$, Lemma \ref{lem:first-order-comm}   and
in addition  for $T_3$  we apply \eqref{eq:quadratic}
with $T=f_{2}$ and \eqref{eq:dobbelt}--\eqref{eq:dobbelt2}.  Clearly \eqref{eq:9} follows from
\eqref{eq:11-}--\eqref{eq:11-b} by one integration.

\noindent{\bf Re \eqref{eq:10}.}  We mimic the proof of \eqref{eq:9}.

Next we note that the above arguments apply to $A\to A-n$ for any
$n\in \Z$ yielding bounds being independent of $n$. Taking
$\epsilon\to 0$ we thus obtain the following bounds for the accretive
operator $T(z)=-\i f_{2}R(z,\sigma)f_{2}$, all being uniform in $n$
and in 
large $ \sigma$ and $\nu$,
\begin{subequations}
  \begin{align*}
    \|\inp{A-n}^{-1}T(z)\inp{A-n}^{-1}\|&\leq \tilde C,\\
    \|F(A<
    n)T(z)\inp{A-n}^{-2}\|&\leq \tilde C,\\
    \|\inp{A-n}^{-2}T(z)F(A\geq n)\|&\leq \tilde C.
  \end{align*}
\end{subequations}

Due to these bounds and Lemmas
\ref{lemma:resolvent-boundsooA}--\ref{lemma:resolvent-boundsooAb} we
conclude \eqref{eq:limitbound3a} with $C=\nobreak16 \tilde C$ provided
$\Im
z>0$ (and hence also if  $\Im
z<0$).

\subsection{Besov
space bound of  resolvent, Proposition \ref{prop:setting-problem} }\label{subsec: Concrete Besov
spaces,Proposition}
We introduce operators
\begin{align*}
  S_\sigma=f_2^{-1}f_1\mand T_\sigma=M(t_\sigma);\;t_\sigma(r)=\sigma
  \parb{\tfrac r{1+r}}^{1/2}(1+r)f^2_1.
\end{align*} Here and henceforth $M(\cdot)$ refers to the operator of
multiplication by the function in the argument. We shall prove the
following lemmas
\begin{lemma}\label{lem:besov-space-boundi} There exists $C>0$
  independent of $\sigma>1$ such that
  \begin{equation}
    \label{eq:S}
    \| S_\sigma v\|_{B(A)}\leq C\| v\|_{B(T_\sigma)}.
  \end{equation}
\end{lemma}

\begin{lemma}\label{lem:besov-space-boundj}
  There exists $C>0$ independent of $\sigma>1$ such that
  \begin{equation}
    \label{eq:f}
    \| f_1^{-1}u\|_{B(T_\sigma)}\leq C\sigma^{1/2} \| u\|_{B(r)}.
  \end{equation}
\end{lemma}

\begin{proof}[Proof of Proposition \ref{prop:setting-problem}]
  We combine Lemmas
  \ref{lem:besov-space-boundi}--\ref{lem:besov-space-boundj} to obtain
  that
  \begin{align*}
    f_2^{-1}=S_\sigma f_1^{-1}\in \vB(B(r),B(A))
  \end{align*} with a bounding constant of the form
  $C\sigma^{1/2}$. Whence, due to Lemma \ref{lemma:resolvent-boundsA},
  \begin{align*}
    R(z,\sigma)=f_2^{-1}\parb{f_2R(z,\sigma)f_2}f_2^{-1}\in
    \vB(B(r),B(r)^*)
  \end{align*} with a bounding constant of the form $C\sigma$.
\end{proof}

\begin{proof}[Proof of Lemma
\ref{lem:besov-space-boundi}]
Since $\|S_\sigma\|\leq 1$ it suffices, due to Lemma
\ref{lemma:resolvent-bounds}, to show the bound
\begin{align}
  \label{eq:bndS}
  \|AS_\sigma v\|\leq C \parb{\|T_\sigma v\|+\|v\|}.
\end{align}

Using \eqref{eq:Afcom} we estimate for all $\psi\in \vD(r)=\vD(M(r))$
\begin{align*}
  \|Af_2^{-1}\psi\|^2&\leq 2\|f_2^{-1}A\psi\|^2+C_1\|f_2^{-1}\psi\|^2\\
  &\leq 4\|f_2^{-1}2p_r r\psi\|^2 +C_2\|f_2^{-1}\psi\|^2\\
  &\leq 16\sigma^2\| r\psi\|^2 +C_2\|S_\sigma f_1^{-1}\psi\|^2\\
  &\leq 16\sigma^2\| \parb{\tfrac r{1+r}}^{1/2}(1+r)f_1\psi\|^2
  +C_2\|f_1^{-1}\psi\|^2.
\end{align*} We apply the estimate to $\psi=f_1v$ yielding
\eqref{eq:bndS} with $C=\max(4, C_2^{1/2})$.
\end{proof}

\begin{proof}[Proof of Lemma \ref{lem:besov-space-boundj}] Introducing
  $\tilde f_1=\sigma^{1/2}f_1$ we need to bound for $j=1,2,3$
  \begin{align}
    \label{eq:fb}
    \| \tilde f_1^{-1}F_ju\|_{B(T_\sigma)}&\leq C \| u\|_{B(r)};\\
    F_1&=F(r<\sigma^{-2/3}),\nonumber\\
    F_2&=F(\sigma^{-2/3}\leq r< 2),\nonumber\\
    F_2&=F(r\geq 2).\nonumber
  \end{align}
  \begin{subequations}
 
    Using that $t_\sigma$ is a bounded function on the support of
    $F_1$ and \eqref{eq:22} we estimate
    \begin{align}
      \label{eq:F1}
      \| \tilde f_1^{-1}F_1u\|_{B(T_\sigma)}\leq
      C_1\sigma^{-1/6}\|u\|\leq C_1\| u\|_{B(r)},
    \end{align} which agrees with \eqref{eq:fb}.

    Let $g_\sigma(r)=\sigma r^{3/2}$ and $G_\sigma=M(g_\sigma)$. Using
    the two-sided estimates $t_\sigma(r)\leq Cg_\sigma(r)$ and
    $g_\sigma (r)\leq Ct_\sigma(r)$, which are valid on the support of
    $F_2$, we can estimate
    \begin{align*}
      \| \tilde f_1^{-1}F_2u\|_{B(T_\sigma)}&\leq C\| \parb{{\sigma
          r}}^{-1/2}F_2u\|_{B(G_\sigma)}\\
      &= C\sum_{2\leq j\leq J}R_j^{1/2}\|F(R_{j-1}\leq
      g_\sigma(r)<R_j)(\sigma
      r^{3/2})^{-1/2}r^{1/4}F_2u\|\\
      &\leq 2^{1/2}C\sum_{2\leq j\leq J}\|F(R_{j-1}\leq
      g_\sigma(r)<R_j)r^{1/4}F_2u\|,
    \end{align*} where $J=J_\sigma\in \N$ is taken smallest such that
    $R_J> 2^{3/2}\sigma$.  By estimating for each term
    \begin{align*}
      r^{1/4}\leq \parb{R_j/\sigma}^{1/6}\leq 2^{(j-J+3)/6},
    \end{align*} we thus obtain
    \begin{align}
      \label{eq:F2}
      \| \tilde f_1^{-1}F_2u\|_{B(T_\sigma)}\leq 2^{1/2}C\sum_{2\leq
        j\leq J}2^{(j-J+3)/6}\|u\|\leq C_1\|u\|_{B(r)},
    \end{align}
    which also agrees with \eqref{eq:fb}.

    Finally using the two-sided estimates $t_\sigma(r)\leq
    C\sigma(1+r)$ and $\sigma(1+r)\leq Ct_\sigma(r)$, which are valid
    on the support of $F_3$, and Lemma \ref{lemma:resolvent-boundsook}
    we can estimate
    \begin{align}
      \begin{split}
        \label{eq:F3}
        \| \tilde f_1^{-1}F_3u\|_{B(T_\sigma)}&\leq
        C_1\sigma^{-1/2}\|u\|_{B(\sigma(1+r))}\\&\leq
        8C_1\|u\|_{B((1+r))}\leq C_2\|u\|_{B(r)},
      \end{split}
    \end{align} which also agrees with \eqref{eq:fb}.

    Having proved \eqref{eq:F1}--\eqref{eq:F3} we conclude
    \eqref{eq:fb}.
  \end{subequations}
\end{proof}

\subsection{Case $\Omega=\bX$}\label{sec:CaseOmegarmc}

We outline a proof of  the analogue of Proposition \ref{prop:setting-problem} for the
case $\Omega=\bX$. This is conceptionally simpler than the previous
case,  and it suffices to mimic parts of the previous proof. We can use the standard conjugate operator
\begin{align}\label{eq:5Abb}
  A=x\cdot p+p\cdot x,
\end{align} rather than the one defined by \eqref{eq:5A}
(alternatively $A$ is given by taking  $r=|x|$ in \eqref{eq:5A}).
We impose
\begin{equation}
      \label{eq:2k2bbB}
      V(x), x\cdot \nabla V(x),\,(x\cdot \nabla)^2 V(x)\in \vC\big (H^1(\bX),H^1(\bX)^*\big ).
    \end{equation} 
We consider a Hilbert space $\vH=L^2(\bX,\d x)\otimes L^2(M,\d
y)$ where interpretation of $x,y$ and $M$ is the same as in Subsection
\ref{subsec: Setting of problem}. Similarly  introducing
$H=\widetilde H^{b}+\widetilde B$ as before the form domains are
\begin{align*}
  Q(\widetilde H^{b})=Q(H)=L^2(M,H^1(\bX);\d y)\subset\vH.
\end{align*} Again we have the property \eqref{eq:banu_stable}. We
define 
\begin{align*}
  f^2=1+\sigma^{-2}p^2,f\geq 0.
\end{align*} We have results similar to Lemma
\ref{lemma:resolvent-boundsA} and Proposition
\ref{prop:setting-problem}.

\begin{lemma}
  \label{lemma:resolvent-boundsA2} With $A$ given by \eqref{eq:5Abb} we
  have uniformly in all large $\sigma,\nu>1$ and all $\Re z\approx 1$
\begin{equation}
    \label{eq:limitbound3a2}
    \left\|
      f R(z,\sigma)
      f \right\|_{\vB (B(A),B(A)^*)} \leq C.
  \end{equation}
\end{lemma}

\begin{proposition}\label{prop:setting-problem2}
  We have uniformly in all large
  $\sigma,\nu>1$ and all $\Re z\approx 1$
  \begin{equation}
    \label{eq:limitbound3a_main2}
    \left\|
      R(z,\sigma)
    \right\|_{\vB (B(|x|),B(|x|)^*)} \leq C\sigma.
  \end{equation}
\end{proposition} 
Given Lemma \ref{lemma:resolvent-boundsA2} we notice that  Proposition
\ref{prop:setting-problem2} is an easy consequence of the following
analogues of Lemmas \ref{lem:besov-space-boundi} and
\ref{lem:besov-space-boundj}. Define
\begin{align*}
     T_\sigma=M(t_\sigma);\;t_\sigma(x,y)=\sigma
  (1+|x|).
  \end{align*}

\begin{lemma}\label{lem:besov-space-boundi2} There exists $C>0$
  independent of $\sigma>1$ such that
  \begin{equation}
    \label{eq:S2}
    \| f^{-1}v\|_{B(A)}\leq C\| v\|_{B(T_\sigma)}.
  \end{equation} 
\end{lemma}

\begin{lemma}\label{lem:besov-space-boundj2}
  There exists $C>0$ independent of $\sigma>1$ such that
  \begin{equation}
    \label{eq:f2}
    \| u\|_{B(T_\sigma)}\leq C\sigma^{1/2} \| u\|_{B(|x|)}.
  \end{equation}
\end{lemma}

We can prove Lemma \ref{lem:besov-space-boundi2} by mimicking the
proof of  Lemma \ref{lem:besov-space-boundi},  while  Lemma
\ref{lem:besov-space-boundj2} is an immediate consequence of Lemma
\ref{lemma:resolvent-boundsook}. 

Whence it remains to show Lemma \ref{lemma:resolvent-boundsA2}. For
that we note the analogue of \eqref{eq:1_comm}  where now $r=|x|$
\begin{align}
    \label{eq:1_comm2}
    \i[H,A]&:=s^{-1}\parb{2p^2 +W}s^{-1}+2\Re\parb{s^{-1}(\nabla
      r^2\cdot \nabla_xs)\widetilde H^b}-\nabla r^2\cdot
    \nabla_x\widetilde B;\\
  W(x)&:=-\nabla r^2\cdot \nabla V(x).\nonumber
\end{align} 

Using \eqref{eq:1_comm2} we can  indeed mimic the proof of Lemma
\ref{lemma:resolvent-boundsA} with $f$ replacing $f_2$.  Note this is
much simpler now. For example there are no factors of $\tfrac r{1+r}$
to consider,  and the analogue of the second commutator is given by
\eqref{eq:secCom} without the second term on the right hand side. We
leave the details to  the reader.

\appendix

\section{}\label{sec:Appendixb} 
In this appendix we show how to undo the commutator $\mathrm{i}[H,A]$.
This is used to obtain (\ref{eq:11.7.13.4.10c}).  Since the
Schr\"odinger operator $H$ is realized with the Dirichlet boundary
condition the approximation procedure of \cite{IS} is not sufficient.

\subsection{Setting}\label{sec:12.6.4.10.14}
We shall work in a generalized setting on a manifold, and present all
conditions needed for the argument independently of the previous
sections.  The case of a constant metric is sufficient for application
to (\ref{eq:11.7.13.4.10c}).  The verification of the conditions below
under the conditions of Sections~\ref{sec:introduction} and
\ref{sec:Reduction to high-energy hard-core sub-system resolvent
  bounds} is straightforward.
  
Let $(\Omega,g)$ be a Riemannian manifold of dimension
$d\ge 1$, and consider the Schr\"odinger operator on ${\mathcal
  H}=L^2(\Omega)=L^2(\Omega,(\det g)^{1/2}\mathrm{d}x)$:
\begin{align*}
  H=H_0+V;\quad H_0=-\tfrac12\Delta=\tfrac12p_i^*g^{ij}p_j,\quad
  p_i=-\mathrm{i}\partial_i.
\end{align*}
We realize $H_0$ as a self-adjoint operator by setting the Dirichlet
boundary condition, i.e.\ $H_0$ is the unique self-adjoint operator
associated with the closure of the quadratic form
\begin{align*}
  \langle H_0\rangle_\psi=\langle \psi, -\tfrac12
  \Delta\psi\rangle,\quad \psi\in C^\infty_{\mathrm{c}}(\Omega).
\end{align*}
We denote the form closure and the self-adjoint realization by the
same symbol $H_0$.  Moreover, we consider the weighted spaces
\begin{align*} {\mathcal H}^s=(H_0+1)^{-s/2}{\mathcal H}, \quad
  s\in\mathbb{R},
\end{align*}
and $H_0$ may also be understood as ${\mathcal H}^s\to{\mathcal
  H}^{s-2}$, $s\in\mathbb{R}$.  For the realization of $H=H_0+V$ we
assume the following condition:
\begin{cond}\label{cond:12.6.2.20.52}
  The potential $V$ is a locally integrable real-valued function, and
  there exist $\varepsilon\in [0,1)$ and $C>0$ such that for any
  $\psi\in C^\infty_{\mathrm{c}}(\Omega)$
  \begin{align*}
    |\langle V\rangle_\psi|\le \varepsilon\langle
    H_0\rangle_\psi+C\|\psi\|^2.
  \end{align*}
\end{cond}
By this condition we extend the form domain of $V$ as ${Q}(V)={\mathcal H}^1$, and this defines a bounded operator $V\colon
{\mathcal H}^1\to{\mathcal H}^{-1}$.  We note, though, this quadratic
form is not necessarily closed.  We henceforth consider $H=H_0+V$ as a
closed quadratic form on $Q(H)={\mathcal H}^1$  or,
equivalently, as a bounded operator ${\mathcal H}^1\to{\mathcal
  H}^{-1}$.  Then the Friedrichs self-adjoint realization of $H$ on
${\mathcal H}$ is the restriction of this $H\colon {\mathcal
  H}^1\to{\mathcal H}^{-1}$ to the domain:
\begin{align*} {\mathcal D}(H)=\{\psi\in{\mathcal H}^1\,|\, H\psi\in
  {\mathcal H}\}\subset {\mathcal H}.
\end{align*}

We next assume a regularity condition for the (virtual) boundary of
$\Omega$:
\begin{cond}\label{cond:12.6.2.21.13}
  There exists a real-valued function $r\in
  C^\infty(\Omega)$  such that:
  \begin{enumerate}
  \item\label{item:12.4.7.19.40} The gradient vector field
    $2\omega=\mathop{\mathrm{grad}} r^2$ on $\Omega$ is complete.
  \item \label{item:1}The following bounds hold:
    \begin{align}
      \sup|\mathrm{d} r|<\infty, \quad
      \sup|\nabla^2 r^2|<\infty,\quad
      \sup_{r\to\infty}|\mathrm{d} \Delta
      r^2|<\infty. \label{eq:10.9.2.23.19}
    \end{align}
  \end{enumerate}
\end{cond}

The function $r$ of Condition~\ref{cond:12.6.2.21.13} is a
generalization of that of previous sections.  For the $r$ of Sections~\ref{sec:introduction} and \ref{sec:Reduction to high-energy
  hard-core sub-system resolvent bounds} we refer to Lemma~\ref{lemma:vector
  field} and Subsection~\ref{Geometric properties} (the completeness
is valid  because the vector field
$\omega$ is tangent to the boundary $\partial \Omega$). For the $r$ of
Subsection \ref{subsec: Setting of problem} we refer to
\eqref{eq:r_bounds}  (the completeness
is valid  because 
$\omega$ vanishes at the boundary $\partial \Omega\times M$).  For the $r$ of
Subsection \ref{sec:CaseOmegarmc} the properties
(\ref{item:12.4.7.19.40}) and (\ref{item:1}) are obvious,  however 
there is a cusp singularity at $x=0$ in this case. A substitute for
Lemmas \ref{lem:12.6.4.13.14}--\ref{lem:12.6.4.3.43}, shown under
Conditions \ref{cond:12.6.2.20.52}--\ref{cond:12.6.2.21.13}, is in
this  case immedidately
provided by the formula  $\|p\mathrm{e}^{\mathrm{i}tA}\psi\|
=\mathrm{e}^{2t}\|p\psi\|$.

By Condition~\ref{cond:12.6.2.21.13} (\ref{item:12.4.7.19.40}) the
vector field $\mathop{\mathrm{grad}} r^2$ generates a one-parameter
group of diffeomorphisms on $\Omega$, which we denote by
\begin{align}
  \mathrm{e}^{2\cdot}\cdot\colon \mathbb{R}\times \Omega\to
  \Omega,\quad (t,x)\mapsto \mathrm{e}^{2t}x.
  \label{eq:12.6.5.2.57}
\end{align}
This satisfies by definition, in local coordinates,
\begin{align}
  \partial_t(\mathrm{e}^{2t}x)^i=g^{ij}(\mathrm{e}^{2t}x)(\partial_jr^2)(\mathrm{e}^{2t}x).
  \label{eq:12.6.5.3.13}
\end{align}
We define the \textit{dilation} $\mathrm{e}^{\mathrm{i}tA}\colon
{\mathcal H}\to{\mathcal H}$ with respect to $r$ by the one-parameter
unitary group
\begin{align*}
  \mathrm{e}^{\mathrm{i}tA}u(x) =J(\mathrm{e}^{2t};x)^{1/2}
  \left(\frac{\det g(\mathrm{e}^{2t}x)}{\det
      g(x)}\right)^{1/4}u(\mathrm{e}^{2t}x),
\end{align*}
where $J$ is the relevant Jacobian.  Note that there is   another
expression:
\begin{align}
  \mathrm{e}^{\mathrm{i}tA}u(x) =\exp \left(\int_0^t\tfrac12(\Delta
    r^2)(\mathrm{e}^{2s}x)\,\mathrm{d}s\right)u(\mathrm{e}^{2t}x).
  \label{eq:12.6.7.1.10}
\end{align}
We let $A$ be the generator of $\mathrm{e}^{\mathrm{i}tA}$.
By the unitarity of $\mathrm{e}^{\mathrm{i}tA}$ the operator $A$ is
self-adjoint, and $C^\infty_{\mathrm{c}}(\Omega)\subseteq {\mathcal
  D}(A)$ is a core for it.  In fact, the dense subspace
$C^\infty_{\mathrm{c}}(\Omega)\subseteq {\mathcal H}$ is invariant
under $\mathrm{e}^{\mathrm{i}tA}$, and for any $u\in
C^\infty_{\mathrm{c}}(\Omega)$ the limit
\begin{align*}
  \lim_{t\to 0}t^{-1}(\mathrm{e}^{\mathrm{i}tA}u-u)
\end{align*}
exists in ${\mathcal H}$.  Note that by (\ref{eq:12.6.7.1.10}) we have
$A$ on $C^\infty_{\mathrm{c}}(\Omega)$ written by
\begin{align*}
  A=\mathrm{i}[H_0,r^2] =\tfrac{1}{2}\{(\partial_i r^2)g^{ij}p_j+p_i^*
  g^{ij}(\partial_j r^2)\} =rp^r+(p^r)^*r,
\end{align*}
where $p^r=-\mathrm{i}\partial^r=-\mathrm{i}(\partial_i
r)g^{ij}\partial_j$.

Let us first consider the commutator $\mathrm{i}[H,A]$ as a quadratic
form defined for $\psi\in C^\infty_{\mathrm{c}}(\Omega)$ by
\begin{align*}
  \langle \mathrm{i}[H,A]\rangle_\psi=\mathrm{i}\langle
  H\psi,A\psi\rangle-\mathrm{i}\langle A\psi,H\psi\rangle.
\end{align*}
In order to discuss its extension we impose the following abstract
form bound condition, which is not quite independent of
Conditions~\ref{cond:12.6.2.20.52} and \ref{cond:12.6.2.21.13}.
\begin{cond}\label{cond:12.6.4.10.16}
  There exists $C>0$ such that for any $\psi\in
  C^\infty_{\mathrm{c}}(\Omega)$
  \begin{align*}
    |\langle \mathrm{i}[H,A]\rangle_\psi|\le C\langle
    H_0+1\rangle_\psi.
  \end{align*}
\end{cond}
Similarly to the above, we henceforth regard $\mathrm{i}[H,A]$ as a
quadratic form on $Q(\mathrm{i}[H,A])={\mathcal H}^1$, 
which may not be closed, or as a bounded operator ${\mathcal
  H}^1\to{\mathcal H}^{-1}$.

\subsection{Preliminaries}\label{sec:12.6.4.6.27}
We prove a regularity property of the flow (\ref{eq:12.6.5.2.57}).
\begin{lemma} \label{lem:12.6.6.20.28} There exists $C>0$ such that
  for any $t\in \mathbb{R}$ and $x\in \Omega$
  \begin{align}
    d\mathrm{e}^{-C|t|} \le
    g^{ij}(x)g_{kl}(\mathrm{e}^{2t}x)[\partial_i(\mathrm{e}^{2t}x)^k][\partial_j(\mathrm{e}^{2t}x)^l]
    \le d \mathrm{e}^{C|t|}.\label{9.12.19.1.58}
  \end{align}
\end{lemma}
\begin{proof}
  The proof is similar to that of \cite[Lemma 2.3]{IS2}.  We note that
  the expression in the middle of (\ref{9.12.19.1.58}) is independent of choice
  of coordinates.  Fix $x\in \Omega$ and choose coordinates such that
  $g_{ij}(x)=\delta_{ij}$.  Consider the vector fields along
  $\{\mathrm{e}^{2t}x\}_{t\in\mathbb{R}}$ given by
  $\partial_i\mathrm{e}^{2t}x$ and $\partial_j\mathrm{e}^{2t}x$.
  Since the Levi-Civita connection $\nabla$ is compatible with the
  metric,
  \begin{align}
    \begin{split}\label{eq:12.6.7.1.34}
      \tfrac{\partial }{\partial
        t}g_{kl}(\mathrm{e}^{2t}x)&[\partial_i(\mathrm{e}^{2t}x)^k][\partial_j(\mathrm{e}^{2t}x)^l]
      =\tfrac{\partial }{\partial
        t}\langle \partial_i\mathrm{e}^{2t}x,
      \partial_j\mathrm{e}^{2t}x\rangle\\
      &=\langle
      \nabla_{\partial_t\mathrm{e}^{2t}x} \partial_i\mathrm{e}^{2t}x,\partial_j\mathrm{e}^{2t}x\rangle
      +\langle \partial_i\mathrm{e}^{2t}x,\nabla_{\partial_t\mathrm{e}^{2t}x} \partial_j\mathrm{e}^{2t}x\rangle.
    \end{split}
  \end{align}
  (The definition of $\nabla_{\partial_t\mathrm{e}^{2t}x}$ is given
  below.)  From (\ref{eq:12.6.5.3.13}) it follows that
  \begin{align*}
    \nabla_{\partial_t \mathrm{e}^{2t}x}\partial_i
    (\mathrm{e}^{2t}x)^\bullet &{}=\partial_t\partial_i
    (\mathrm{e}^{2t}x)^\bullet
    +[\partial_t (\mathrm{e}^{2t}x)^k]\Gamma^\bullet_{kl}\partial_i (\mathrm{e}^{2t}x)^l\\
    &{}=\partial_i\partial_t (\mathrm{e}^{2t}x)^\bullet
    +( g^{km}\partial_mr^2)\Gamma^\bullet_{kl}\partial_i (\mathrm{e}^{2t}x)^l\\
    &{}=[\partial_i (\mathrm{e}^{2t}x)^k]\partial_k(g^{\bullet
      l}\partial_lr^2)
    +[\partial_i (\mathrm{e}^{2t}x)^l]\Gamma^\bullet_{kl}g^{km}\partial_mr^2\\
    &{}=\nabla_{\partial_i \mathrm{e}^{2t}x}(g^{\bullet l}\partial_l r^2)\\
    &{}=g^{\bullet l}[\partial_i (\mathrm{e}^{2t}x)^k] (\nabla^2
    r^2)_{kl}.
  \end{align*}
  Thus, plugging this into (\ref{eq:12.6.7.1.34}) and taking a
  contraction with $g^{ij}(x)=\delta^{ij}$, we obtain
  \begin{align*}
    \Bigl|\tfrac{\partial }{\partial
      t}g^{ij}(x)g_{kl}(\mathrm{e}^{2t}x)[\partial_i(\mathrm{e}^{2t}x)^k]
    [\partial_j(\mathrm{e}^{2t}x)^l]\Bigr| &\le C
    g^{ij}(x)g_{kl}(\mathrm{e}^{2t}x)[\partial_i(\mathrm{e}^{2t}x)^k]
    [\partial_j(\mathrm{e}^{2t}x)^l].
  \end{align*}
  Noting
  $g^{ij}(x)g_{kl}(\mathrm{e}^{2t}x)[\partial_i(\mathrm{e}^{2t}x)^k]
  [\partial_j(\mathrm{e}^{2t}x)^l]\bigr|_{t=0}=d$, we have
  (\ref{9.12.19.1.58}).
\end{proof}

Recall the functions $\chi_\nu,\bar\chi_\nu\in C^\infty(\mathbb{R})$
of 
Subsubsection~\ref{subsec:Notation}.  We shall henceforth consider the
functions $\chi_\nu=\chi_\nu(r)$, $\bar\chi_\nu=\bar\chi_\nu(r)$ as
being composed with the function $r$ from
Condition~\ref{cond:12.6.2.21.13}.  We also set
\begin{align*}
  \chi_{\nu,\nu'}=\chi_\nu\bar\chi_{\nu'},\quad
  \bar\chi_{\nu'}=1-\chi_{\nu'},\quad \nu'\ge 2\nu\ge 2.
\end{align*}
Next, we prove the following statement:
\begin{lemma}\label{lem:11.7.19.22.23}
  Let $\psi\in {\mathcal D}(H)$. Then there exists $\nu_0>0$ such
  that, for any $\nu> \nu_0$ and any $\sigma\ge 0$ with
  $\mathrm{e}^{\sigma r}\psi,\mathrm{e}^{\sigma r}H\psi \in {\mathcal
    H}$, one has $\mathrm{e}^{\sigma r} \chi_\nu\psi\in {\mathcal
    D}(H)$.
\end{lemma}
\begin{proof} {\noindent \it Step I.}
  We first claim $\mathrm{e}^{\sigma r}\chi_{\nu,\nu'}\psi\in
  {\mathcal D}(H)$.  Since $\psi\in{\mathcal H}^1$, we have
  \begin{align*}
    \mathrm{e}^{\sigma r}\chi_{\nu,\nu'}\psi, \mathrm{e}^{\sigma
      r}\chi_{\nu,\nu'} p\psi\in {\mathcal H},
  \end{align*}
  and hence $p\mathrm{e}^{\sigma r}\chi_{\nu,\nu'} \psi\in {\mathcal
    H}$ by (\ref{eq:10.9.2.23.19}).  Choose a sequence $\psi_n\in
  C^\infty_{\mathrm{c}}(\Omega)$ such that, as $n\to\infty$,
  \begin{align}
    \|\psi-\psi_n\|+\|p(\psi-\psi_n)\|\to 0,
    \label{eq:12.4.7.21.56}
  \end{align}
  and then by (\ref{eq:10.9.2.23.19}) again, as $n\to\infty$,
  \begin{align*}
    \mathrm{e}^{\sigma r}\chi_{\nu,\nu'}\psi_n\to \mathrm{e}^{\sigma
      r}\chi_{\nu,\nu'}\psi, \quad p\mathrm{e}^{\sigma
      r}\chi_{\nu,\nu'}\psi_n\to p\mathrm{e}^{\sigma
      r}\chi_{\nu,\nu'}\psi \quad \mbox{ in }{\mathcal H}.
  \end{align*}
  This implies that $\mathrm{e}^{\sigma r}\chi_{\nu,\nu'}\psi\in
  {\mathcal H}^1$.  Note the distributional identity
  \begin{align*}
    H \mathrm{e}^{\sigma r}\chi_{\nu,\nu'}\psi =\mathrm{e}^{\sigma
      r}\chi_{\nu,\nu'}H\psi -\mathrm{e}^{\sigma r}(\sigma
    \chi_{\nu,\nu'}+\chi_{\nu,\nu'}')\partial^r\psi
    -\tfrac12(\Delta \mathrm{e}^{\sigma r}\chi_{\nu,\nu'})\psi.
  \end{align*}
  Then, since $\psi, p\psi, H\psi\in {\mathcal H}$ and by
  (\ref{eq:10.9.2.23.19})
  \begin{align}
    \chi_\nu |\Delta r|=\tfrac{1}{2r}\chi_\nu|(\Delta
    r^2)-2|\mathrm{d}r|^2|\le C_\nu,
    \label{eq:11.7.22.9.52}
  \end{align}
  we have $H\mathrm{e}^{\sigma r}\chi_{\nu,\nu'}\psi\in {\mathcal H}$.
  Hence $\mathrm{e}^{\sigma r}\chi_{\nu,\nu'}\psi\in {\mathcal D}(H)$.

  \smallskip {\noindent \it Step II.}  We next show
  $\mathrm{e}^{\sigma r}\chi_{\nu}p\psi\in {\mathcal H}$.  Noting that
  $\mathrm{e}^{\sigma r}\chi_{\nu,\nu'}\psi \in {\mathcal H}^1$ as in
  Step I, we estimate by Condition~\ref{cond:12.6.2.20.52}
  \begin{align*}
    \|\mathrm{e}^{\sigma r}\chi_{\nu,\nu'}p\psi\|^2 &\le
    2\|p\mathrm{e}^{\sigma r}\chi_{\nu,\nu'}\psi\|^2
    +2\|(p\mathrm{e}^{\sigma r}\chi_{\nu,\nu'})\psi\|^2\\
    &\le C\langle H\rangle_{\mathrm{e}^{\sigma r}\chi_{\nu,\nu'}\psi}
    +C_{\sigma}\|\mathrm{e}^{\sigma r}\chi_{\nu/2,2\nu'}\psi\|^2.
  \end{align*}
  Then, by $\mathrm{e}^{\sigma r}\chi_{\nu,\nu'}\psi \in {\mathcal
    D}(H)$ and the definition of the Friedrichs extension
  \begin{align*}
    \|\mathrm{e}^{\sigma r}\chi_{\nu,\nu'}p\psi\|^2 \le {}& C\langle
    \mathrm{e}^{\sigma r}\chi_{\nu,\nu'}\psi,H\mathrm{e}^{\sigma
      r}\chi_{\nu,\nu'}\psi\rangle
    +C_{\sigma}\|\mathrm{e}^{\sigma r}\chi_{\nu/2,2\nu'}\psi\|^2\\
    \le {}& C\langle \mathrm{e}^{\sigma
      r}\chi_{\nu,\nu'}\psi,\mathrm{e}^{\sigma
      r}\chi_{\nu,\nu'}H\psi\rangle
    +C\langle \mathrm{e}^{\sigma r}\chi_{\nu,\nu'}\psi,[-\tfrac 12\Delta,\mathrm{e}^{\sigma r}\chi_{\nu,\nu'}] \psi\rangle\\
    &{}+C_{\sigma}\|\mathrm{e}^{\sigma r}\chi_{\nu/2,2\nu'}\psi\|^2\\
    \le {}& C\|\mathrm{e}^{\sigma r}\chi_{\nu,\nu'}H \psi\|^2
    +\tfrac12 \|\mathrm{e}^{\sigma r}\chi_{\nu,\nu'}p\psi\|^2
    +C_{\sigma}\|\mathrm{e}^{\sigma r}\chi_{\nu/2,2\nu'}\psi\|^2.
  \end{align*}
  Hence
  \begin{align*}
    \|\mathrm{e}^{\sigma r}\chi_{\nu,\nu'}p\psi\|^2 \le
    C\|\mathrm{e}^{\sigma r}\chi_{\nu,\nu'}H \psi\|^2 +C_{\sigma}
    \|\mathrm{e}^{\sigma r}\chi_{\nu/2,2\nu'}\psi\|^2.
  \end{align*}
  Now we let $\nu'\to\infty$. Then by the Lebesgue convergence theorem
  $\mathrm{e}^{\sigma r}\chi_{\nu}p\psi\in {\mathcal H}$.

  \smallskip {\noindent \it Step III.}  We note $p\mathrm{e}^{\sigma
    r}\chi_{\nu}\psi\in {\mathcal H}$ by Step II.  We choose a
  sequence $\psi_n\in C^\infty_{\mathrm{c}}(\Omega)$ satisfying
  (\ref{eq:12.4.7.21.56}) as $n\to\infty$, and estimate
  \begin{align}
    \|\mathrm{e}^{\sigma r} \chi_\nu\psi-\mathrm{e}^{\sigma r}
    \chi_{\nu,\nu'}\psi_n\| +\|p(\mathrm{e}^{\sigma r}
    \chi_\nu\psi-\mathrm{e}^{\sigma r} \chi_{\nu,\nu'}\psi_n)\|.
    \label{eq:12.4.8.22.50}
  \end{align}
  For $\nu'\ge 2\nu$ we have the first term of (\ref{eq:12.4.8.22.50})
  bounded by
  \begin{align*}
    \|\mathrm{e}^{\sigma r} \chi_\nu\psi-\mathrm{e}^{\sigma r}
    \chi_{\nu,\nu'}\psi_n\| \le \|\mathrm{e}^{\sigma r}
    \chi_{\nu'}\psi\| +\|\mathrm{e}^{\sigma r}
    \chi_{\nu,\nu'}(\psi-\psi_n)\|,
  \end{align*}
  and the second term bounded by
  \begin{align*}
    &\|p(\mathrm{e}^{\sigma r} \chi_\nu\psi-\mathrm{e}^{\sigma r} \chi_{\nu,\nu'}\psi_n)\|\\
    &\le \|p\mathrm{e}^{\sigma r} \chi_{\nu'}\psi\|
    +\|p\mathrm{e}^{\sigma r} \chi_{\nu,\nu'}(\psi-\psi_n)\|\\
    &\le \|\mathrm{e}^{\sigma r} \chi_{\nu'} p \psi\| + C_\sigma
    \|\mathrm{e}^{\sigma r} \chi_{\nu'/2,2\nu'} \psi\|
    +\|\mathrm{e}^{\sigma r} \chi_{\nu,\nu'} p(\psi-\psi_n)\|
    +C_\sigma\|\mathrm{e}^{\sigma r}
    \chi_{\nu'/2,2\nu'}(\psi-\psi_n)\|.
  \end{align*}
  Thus we can make (\ref{eq:12.4.8.22.50}) arbitrarily small by
  letting $\nu'$ be large and then $n$ large.  Then we obtain a
  sequence verifying
  \begin{align*}
    \|\mathrm{e}^{\sigma r} \chi_\nu\psi-\mathrm{e}^{\sigma r}
    \chi_{\nu,\nu'(m)}\psi_{n(m)}\| +\|p(\mathrm{e}^{\sigma r}
    \chi_\nu\psi-\mathrm{e}^{\sigma r}
    \chi_{\nu,\nu'(m)}\psi_{n(m)})\|\to 0
  \end{align*}
  as $m\to\infty$, and hence $\mathrm{e}^{\sigma r} \chi_\nu\psi\in
  {\mathcal H}^1$.

  Finally, noting the distributional identity
  \begin{align*}
    H\mathrm{e}^{\sigma r} \chi_\nu\psi =\mathrm{e}^{\sigma r}
    \chi_\nu H\psi +[-\tfrac12\Delta ,\mathrm{e}^{\sigma r}
    \chi_\nu]\psi,
  \end{align*}
  we learn $H\mathrm{e}^{\sigma r} \chi_\nu\psi\in {\mathcal H}$, and
  hence $\mathrm{e}^{\sigma r} \chi_\nu\psi\in {\mathcal D}(H)$.
\end{proof}

\begin{corollary}\label{cor:12.6.9.11.44}
  Let $\psi\in {\mathcal D}(H)$ satisfy $\mathrm{e}^{\sigma
    r}\psi,\mathrm{e}^{\sigma r}H\psi \in {\mathcal H}$ for any
  $\sigma\ge 0$.  Then for large $\nu>0$ one has $\mathrm{e}^{\sigma
    r} \chi_\nu\psi\in {\mathcal D}(H)\cap {\mathcal D}(A)$.
\end{corollary}

\subsection{Undoing commutators}\label{sec:12.6.4.7.0}

\begin{lemma}\label{lem:12.6.4.13.14} 
  For any $s\in [-1,1]$ the inclusion
  $\mathrm{e}^{\mathrm{i}tA}{\mathcal H}^s\subseteq{\mathcal H}^s$
  holds, and
  \begin{align}
    \sup_{|t|<1} \|\mathrm{e}^{\mathrm{i}tA}\|_{{\mathcal B}({\mathcal
        H}^s)}<\infty.
    \label{eq:12.6.7.2.34}
  \end{align}
  Moreover, $\mathrm{e}^{\mathrm{i}tA}\colon {\mathcal H}^s\to{\mathcal H}^s$ is strongly continuous in $t\in \mathbb{R}$.
\end{lemma}
\begin{proof}
  Let us first set $s=1$.  For any $\psi \in
  C^\infty_{\mathrm{c}}(\Omega)$ we can compute by
  (\ref{eq:12.6.7.1.10})
  \begin{align}
    \begin{split}
    &p_i(\mathrm{e}^{\mathrm{i}tA}\psi)(x)\\
    &=\left(\int_0^t\tfrac12
      [p_i(\mathrm{e}^{2s}x)^j](\partial_j\Delta
      r^2)(\mathrm{e}^{2s}x)\,\mathrm{d}s \right)
    (\mathrm{e}^{\mathrm{i}tA}\psi)(x)
    +[\partial_i(\mathrm{e}^{2t}x)^j](\mathrm{e}^{\mathrm{i}tA}p_j\psi)(x),
  \end{split}\label{eq:12.7.25.2.39}
  \end{align}
  where we abuse $\mathrm{e}^{\mathrm{i}tA}$ as if it were defined for
  the sections of the cotangent bundle.  Then by
  (\ref{eq:10.9.2.23.19}) and Lemma~\ref{lem:12.6.6.20.28} for any
  $|t|\le T$
  \begin{align*}
    \|\mathrm{e}^{\mathrm{i}tA}\psi\|_{{\mathcal H}^1}^2
    &=\|\psi\|_{{\mathcal H}}^2
    +\|p\mathrm{e}^{\mathrm{i}tA}\psi\|_{{\mathcal H}}^2\\
    &\le \|\psi\|_{{\mathcal H}}^2
    +C_T\|\mathrm{e}^{\mathrm{i}tA}\psi\|_{{\mathcal H}}^2
    +C_T\|\mathrm{e}^{\mathrm{i}tA} p\psi\|_{{\mathcal H}}^2\\
    &\le C_T\|\psi\|_{{\mathcal H}^1}^2.
  \end{align*}
  By the density argument this implies
  $\mathrm{e}^{\mathrm{i}tA}{\mathcal H}^1\subseteq{\mathcal H}^1$,
  and moreover for any $\psi\in{\mathcal H}^1$ and $|t|\le T$
  \begin{align*}
    \|\mathrm{e}^{\mathrm{i}tA}\psi\|_{{\mathcal H}^1}^2\le
    C_T\|\psi\|_{{\mathcal H}^1}^2.
  \end{align*}
  Thus the uniform boundedness principle applies and
  (\ref{eq:12.6.7.2.34}) follows for $s=1$.  
  As for the strong continuity as ${\mathcal H}^1\to{\mathcal H}^1$, 
  we can show it first on $C^\infty_{\mathrm{c}}(\Omega)$ using (\ref{eq:12.7.25.2.39}) and 
  $\partial_i(\mathrm{e}^{2t}x)^j=\delta_i^j$,
  and then extend it by the boundedness.
  We can show the same
  result for $s=-1$ by taking the adjoint, and then the assertion is
  proved for $s\in [-1,1]$ by  interpolation.
\end{proof}

\begin{lemma}\label{lem:12.6.3.18.30}
  There exists $C>0$ such that for any $|t|<1$
  \begin{align*}
    \|H\mathrm{e}^{\mathrm{i}tA}-\mathrm{e}^{\mathrm{i}tA}H\|_{{\mathcal
        B}({\mathcal H}^{1},{\mathcal H}^{-1})}\le C|t|
  \end{align*}
\end{lemma}
\begin{proof}
  As a quadratic form on $C^\infty_{\mathrm{c}}(\Omega)$, or as an
  operator $C^\infty_{\mathrm{c}}(\Omega)\to{\mathcal H}^{-1}$, 
  \begin{align*}
    H\mathrm{e}^{\mathrm{i}tA}-\mathrm{e}^{\mathrm{i}tA}H
    &=\int_0^t \tfrac{\mathrm{d}}{\mathrm{d}s} \mathrm{e}^{\mathrm{i}(t-s)A}H\mathrm{e}^{\mathrm{i}sA}\,\mathrm{d}s\\
    &=\int_0^t
    \mathrm{e}^{\mathrm{i}sA}\mathrm{i}[H,A]\mathrm{e}^{\mathrm{i}(t-s)A}\,\mathrm{d}s.
  \end{align*}
  Then by Lemma \ref{lem:12.6.4.13.14} and the density argument of 
  $C^\infty_{\mathrm{c}}(\Omega)\subseteq {\mathcal H}^{1}$ the
  assertion follows.
\end{proof}

\begin{lemma}\label{lem:12.6.4.3.43} 
  The following strong limit to the right exists in ${\mathcal
    B}({\mathcal H}^{1},{\mathcal H}^{-1})$, and the following equality holds
  \begin{align}
    \mathrm{i}[H,A] =\slim_{t\to 0}
    t^{-1}[H\mathrm{e}^{\mathrm{i}tA}-\mathrm{e}^{\mathrm{i}tA}H].
    \label{eq:12.6.4.3.29}
  \end{align}
\end{lemma}
\begin{proof}
  For any $\psi\in C^\infty_{\mathrm{c}}(\Omega)$
  \begin{align*}
    t^{-1}(H\mathrm{e}^{\mathrm{i}tA}-\mathrm{e}^{-\mathrm{i}tA}H)\psi-\mathrm{i}[H,A]\psi
    =t^{-1}\int_0^t \bigl\{\mathrm{e}^{\mathrm{i}sA}\mathrm{i}[H,A]\mathrm{e}^{\mathrm{i}(t-s)A}-\mathrm{i}[H,A]\bigr\}
     \psi\,\mathrm{d}s.
  \end{align*}
  We use the strong continuity of $\mathrm{e}^{\mathrm{i}tA}$ of Lemma~\ref{lem:12.6.4.13.14} to obtain 
  (\ref{eq:12.6.4.3.29}) on $C^\infty_{\mathrm{c}}(\Omega)$.
  Then by Lemma~\ref{lem:12.6.3.18.30} and the density
  argument, the strong limit of (\ref{eq:12.6.4.3.29}) exists in
  ${\mathcal B}({\mathcal H}^{1},{\mathcal H}^{-1})$.  
\end{proof}

The following lemma is a main goal of this appendix:
\begin{lemma}
  Let $\psi\in {\mathcal D}(H)$ satisfy $\mathrm{e}^{\sigma
    r}\psi,\mathrm{e}^{\sigma r}H\psi\in {\mathcal H}$ for any
  $\sigma\ge 0$.  Then,
  \begin{align*}
    \langle \mathrm{i}[H,A]\rangle_{\mathrm{e}^{\sigma r}\chi_\nu\psi}
    =\mathrm{i}\langle H\mathrm{e}^{\sigma r}\chi_\nu\psi,
    A\mathrm{e}^{\sigma r}\chi_\nu\psi\rangle - \mathrm{i}\langle
    A\mathrm{e}^{\sigma r}\chi_\nu\psi, H\mathrm{e}^{\sigma
      r}\chi_\nu\psi\rangle.
  \end{align*}
\end{lemma}
\begin{proof}
  We note $\mathrm{e}^{\sigma r}\chi_\nu\psi\in {\mathcal D}(H)\cap
  {\mathcal D}(A)$ by Corollary~\ref{cor:12.6.9.11.44}.  Then, by
  Lemma~\ref{lem:12.6.4.3.43}
  \begin{align*}
    \langle \mathrm{i}[H,A]\rangle_{\mathrm{e}^{\sigma r}\chi_\nu\psi}
    &=\lim_{t\to 0}
    \langle t^{-1}[H\mathrm{e}^{\mathrm{i}tA}-\mathrm{e}^{\mathrm{i}tA}H]\rangle_{\mathrm{e}^{\sigma r}\chi_\nu\psi}\\
    &=\mathrm{i}\langle H\mathrm{e}^{\sigma r}\chi_\nu\psi,
    A\mathrm{e}^{\sigma r}\chi_\nu\psi\rangle - \mathrm{i}\langle
    A\mathrm{e}^{\sigma r}\chi_\nu\psi, H\mathrm{e}^{\sigma
      r}\chi_\nu\psi\rangle.
  \end{align*}
\end{proof}

Finally we examine \eqref{eq:banu_stable} and \eqref{eq:banu_stableb}.
It is easy to prove \eqref{eq:banu_stable} by Lemma~\ref{lem:12.6.4.13.14} for $s=1$ combined with 
the smallness of $V$, cf.\ Condition~\ref{cond:12.6.2.20.52}.
The bound \eqref{eq:banu_stableb} is already obtained by the explicit formula $\|p_r\mathrm{e}^{\mathrm{i}tA}\psi\|
=\mathrm{e}^{2t}\|p_r\psi\|$, but we show the corresponding slightly general statement under the 
setting of the appendix for the quadratic form  $P_r:=\{p^r+(p^r)^*\}^2$.
We consider $P_r$ as a closed form with 
domain  $Q(P_r)=\vD(p^r)$.
\begin{lemma}\label{lem:12.6.4.13.14b}
  The inclusion
  $\mathrm{e}^{\mathrm{i}tA}Q(P_r)\subseteq Q(P_r)$
  holds, and
  \begin{align}
    \sup_{|t|<1} \|\mathrm{e}^{\mathrm{i}tA}\|_{{\mathcal B}(Q(P_r))}<\infty.
    \label{eq:12.6.7.2.34b}
  \end{align}
\end{lemma}
\begin{proof}
We first note that 
\begin{align*}
\partial_t r(\mathrm{e}^{2t}x)=2r(\mathrm{e}^{2t}x) (\partial^rr)(\mathrm{e}^{2t}x)
=2r(\mathrm{e}^{2t}x) |\mathrm{d}r(\mathrm{e}^{2t}x)|^2,
\end{align*}
and this implies
\begin{align}
r(\mathrm{e}^{2t}x)=r(x)\,\exp \left(2\int_0^t |\mathrm{d}r(\mathrm{e}^{2s}x)|^2\,\mathrm{d}s\right).
\label{eq:12.6.27.0.41}
\end{align}
Let $\psi\in C^\infty_{\mathrm{c}}(M)$.
Then, by $A\mathrm{e}^{\mathrm{i}tA}\psi=\mathrm{e}^{\mathrm{i}tA}A\psi$
and $A=r(p^r+(p^r)^*)+\tfrac1{\mathrm{i}}|\mathrm{d}r|^2$,
we can compute
\begin{align*}
&r(x)((p^r+(p^r)^*)\mathrm{e}^{\mathrm{i}tA}\psi)(x)\\
&=(\mathrm{e}^{\mathrm{i}tA}r(p^r+(p^r)^*)\psi)(x)
+\tfrac{1}{\mathrm{i}}(|\mathrm{d}r(\mathrm{e}^{2t}x)|^2-|\mathrm{d}r(x)|^2)(\mathrm{e}^{\mathrm{i}tA}\psi)(x)\\
&=r(\mathrm{e}^{2t}x)(\mathrm{e}^{\mathrm{i}tA}(p^r+(p^r)^*)\psi)(x)
+\tfrac{1}{\mathrm{i}}\left[\int_0^t 
\partial_s|\mathrm{d}r(\mathrm{e}^{2s}x)|^2\,\mathrm{d}s\right] (\mathrm{e}^{\mathrm{i}tA}\psi)(x)\\
&=r(\mathrm{e}^{2t}x)(\mathrm{e}^{\mathrm{i}tA}(p^r+(p^r)^*)\psi)(x)
+\tfrac{2}{\mathrm{i}}\left[\int_0^t 
r(\mathrm{e}^{2s}x)(\partial^r|\mathrm{d}r|^2)(\mathrm{e}^{2s}x)\,\mathrm{d}s\right]
(\mathrm{e}^{\mathrm{i}tA}\psi)(x),
\end{align*}
so that we obtain using (\ref{eq:12.6.27.0.41})
\begin{align}
|((p^r+(p^r)^*)\mathrm{e}^{\mathrm{i}tA}\psi)(x)|
\le C_T\bigl[|(\mathrm{e}^{\mathrm{i}tA}(p^r+(p^r)^*)\psi)(x)|
+|(\mathrm{e}^{\mathrm{i}tA}\psi)(x)|\bigr]
\label{eq:12.6.27.0.47}
\end{align}
for $|t|<T$ and $x\notin r^{-1}(0)$.
For $x\in r^{-1}(0)$ it is obvious that $\mathrm{e}^{2t}x=x$, and 
hence (\ref{eq:12.6.27.0.47}) holds also in the interior of the closed set $r^{-1}(0)$.
Since both sides of (\ref{eq:12.6.27.0.47}) are continuous, we have (\ref{eq:12.6.27.0.47})
in $\Omega$ and it follows that 
\begin{align*}
\|\mathrm{e}^{\mathrm{i}tA}\psi\|_{Q(P_r)}^2\le C_T\|\psi\|_{Q(P_r)}^2.
\end{align*}
By the density argument and the uniform boundedness theorem the assertion follows.
\end{proof}

\section{}\label{sec:Appendix} 
In this appendix we introduce the notion of strictly convexity of an
obstacle and derive the geometric properties needed for the one-body
type model considered in Subsection \ref{subsec: Setting of problem}.

Let $\Theta\subset \mathbb{R}^d$, $d\ge 2$, be
a bounded open set, denote its
closure by $\overline\Theta$, and set $\Omega=\mathbb{R}^d\setminus
\overline{\Theta}$.  The goal of these short notes is to give a
criterion for the existence of a function $r\in C^\infty(\Omega)$ such
that for some $c>0$ 
\begin{subequations}
  \begin{align}
    |\nabla r|&=1\mbox{ in }\Omega,\label{eq:12.5.27.14.13}\\
    (\nabla^2 r)_{|S_r}&\ge c\langle r\rangle^{-1}g_{|S_r},\label{eq:12.5.27.14.14}\\
    |\partial^\alpha r|&\le C_\alpha \langle
    r\rangle^{1-|\alpha|},\label{eq:12.5.27.14.15}
  \end{align}
\end{subequations}
where $g$ is the Euclidean metric,
$S_{r}=r^{-1}(r)$ is the level surface and $g_{|S_r}$ is the
pull-back of $g$ to $S_r$. Note that $S_r$ is
smooth by (\ref{eq:12.5.27.14.13}).  
 We impose
the following convexity type condition for $\Theta$.  Note that the
inequality (\ref{eq:12.5.27.14.14}) represents the convexity of $r$.

\begin{cond}\label{cond:12.5.27.15.10}
  Let $\Theta\subset \mathbb{R}^d$, $d\ge 2$, be an open connected
  subset with smooth boundary $S=\partial \Theta$, and
  $\nu\in\Gamma(N^+S)$ be the outward unit normal vector field on $S$.
  There exists a constant $\kappa> 0$ such that
  \begin{align}
    (\nabla \nu)_{|S}\ge \kappa g_{|S}.
    \label{eq:12.5.27.15.29}
  \end{align}
\end{cond}

A subset $\Theta\subset \mathbb{R}^d$ fulfilling Condition
\ref{cond:12.5.27.15.10} is called {\it strictly convex}. We show in
Lemma \ref{lemma:convex} that such set is convex.

\begin{remarks}
  \begin{enumerate}
  \item The symmetric tensor $(\nabla \nu)_{|S}$ is called the
    \textit{second fundamental form} of $S$, and its eigenvalues
    relative to $g_{|S}$ are called the \textit{principal curvatures}.
    Hence (\ref{eq:12.5.27.15.29}) implies that the principal
    curvatures are bounded below by $\kappa> 0$.  For these notions we
    refer to \cite[Section II.2]{C},
    although we adopt different signs for them.

  \item The Bonnet-Mayers theorem, cf. 
    \cite[Theorem II.6.1]{C}, applies to $S$, and, in particular, $S$
    is compact.  In fact, the sectional curvatures are the products of
    two principal curvatures when the hypersurface under consideration
    is embedded into the Euclidean space, cf. \cite[Theorem II.2.1]{C}, and thus the Ricci
    curvature is bounded below by a positive constant.  See also
    \cite[Exercise II.6]{C}.  By the same fact we can rewrite
    (\ref{eq:12.5.27.15.29}) in terms of the intrinsic sectional
    curvatures of $S$, assuming the compactness of $S$.  We note,
    however, that for these arguments we need $d\ge 3$ to have
    nontrivial sectional curvatures.

  \end{enumerate}
\end{remarks}

\begin{proposition}\label{thm:12.5.27.16.50}
  Suppose $\Theta\subset \mathbb{R}^d$ is strictly convex.  Then the
  distance function $r(x)=\mathop{\mathrm{dist}}(x,\Theta)$, $x\in
  \Omega$, satisfies
  (\ref{eq:12.5.27.14.13})--(\ref{eq:12.5.27.14.15}).
\end{proposition}

In the sequel we prove Proposition~\ref{thm:12.5.27.16.50}.  We first
show the convexity of $\Theta$.
\begin{lemma}\label{lem:12.5.29.1.20}
  Let $x\in S$ and set $\tilde{S}_x=\exp [(TS)_x]$.  Then there exists
  a neighborhood $U$ of $x$ in $\mathbb{R}^d$ such that
  $\overline{\Theta}\cap \tilde{S}_x\cap U=\{x\}$.
\end{lemma}
This is a sort of \textit{local convexity}. We omit the proof, just
referring to  \cite[Exercise II.4]{C}.
\begin{lemma}\label{lemma:convex}
  For any $x,y\in \Theta$ the geodesic $\gamma_{xy}$ connecting $x$
  and $y$ lies in $\Theta$.
\end{lemma}
\begin{proof}
  Let us assume otherwise, i.e.\ assume that the set
  \begin{align*}
    \Phi=\{(x,y)\in \Theta\times \Theta\,|\, \gamma_{xy}([0,1])\subset
    \Theta\}
  \end{align*}
  does not coincide with $\Theta\times \Theta$.  Since
  $\gamma_{xy}(t)$ is continuously dependent on $(t,x,y)$, it is clear
  that $\Phi$ is open in $\Theta\times \Theta$.  Choose
  $(x_0,y_0)\in \partial\Phi\subset \Theta\times \Theta$, and then by
  definition of $\Phi$ we have a proper open subset
  $\gamma_{x_0y_0}^{-1}(\Theta)\subsetneq [0,1]$.  Since any open
  subset of $[0,1]$ is a union of at most countable number of open
  intervals, we can find points $x_1,y_1\in \gamma_{x_0y_0}([0,1])\cap
  \Theta$ such that the number of component of
  $\gamma_{x_1y_1}^{-1}(\Theta)$ is equal to $2$.  We note we still
  have $(x_1,y_1)\in \partial\Phi$, and hence
  \begin{align}
    \gamma_{x_1y_1}^{-1}(\Theta)=[0,\tau)\cup (\tau,1]
    \label{eq:12.5.29.2.48}
  \end{align}
  for some $\tau\in (0,1)$.  In fact, if
  $\gamma_{x_1y_1}^{-1}(\Theta)=[0,\tau_1)\cup (\tau_2,1]$,
  $0<\tau_1<\tau_2<1$, we have to have
  $\gamma_{x_1y_1}([\tau_1,\tau_2])\subset S$ due to
  $(x_1,y_1)\in \partial\Phi$ and this contradicts
  Lemma~\ref{lem:12.5.29.1.20}.  But (\ref{eq:12.5.29.2.48}) also
  contradicts Lemma~\ref{lem:12.5.29.1.20}, and thus
  $\Phi=\Theta\times\Theta$.
\end{proof}

Now we are ready to give the distorted spherical coordinates for
$\Omega$.
\begin{lemma}\label{lem:12.5.27.17.21}
  The exponential map on the outward normal vectors on $S$:
  \begin{align*}
    \exp_{|N^+S} \colon N^+S\to \Omega
  \end{align*}
  is bijective.
\end{lemma}
\begin{proof}
  We shall intensively use the convexity of $\overline\Theta$.  Let us
  denote an element of $N^+S$ by $r\nu(\sigma)$, $(r,\sigma)\in
  (0,\infty)\times S$.  If $\exp(r\nu(\sigma))\in \overline{\Theta}$
  for some $(r,\sigma)$, then by the convexity this contradicts the
  fact that $\nu$ is outward.  Thus the image $\exp(N^+S)$ is included
  in $\Omega$.

  Next, assume $\exp(r\nu(\sigma))=\exp(r'\nu(\sigma'))$ for some
  $(r,\sigma)$, $(r',\sigma')$.  By the convexity we note that the
  obstacle $\Theta$ is in one side of the half space devided by the
  tangent plane at $\sigma$, and by the normality of $\nu(\sigma)$ we
  obtain $\mathop{\mathrm{dist}}(\exp(r\nu(\sigma)),\Theta)=r$.  Thus
  $r=r'$. Moreover, by the convexity of $\overline\Theta$ and the
  minimality of
  $r=r'=\mathop{\mathrm{dist}}(\exp(r\nu(\sigma)),\Theta)$,
  $\sigma=\sigma'$.

  Finally take any $x\in \Omega$, and then we can find $y\in \partial
  \Theta$ such that $\mathop{\mathrm{dist}}(x,\Theta)=\|x-y\|$.  Then
  the geodesic connecting $x$ and $y$ is orthogonal to $\partial
  \Theta$, because, otherwise, $\|x-y\|$ does not give a minimal
  distance. This implies $x$ is in the image $\exp(N^+S)$.
\end{proof}
As in the proof above we identify $N^+S\cong (0,\infty)\times S$
through
\begin{align*}
  N^+S\ni r\nu(\sigma)\leftrightarrow (r,\sigma)\in (0,\infty)\times
  S,
\end{align*}
and consider $(r,\sigma)\in (0,\infty)\times S$ as local coordinates
of $N^+S$.  By Lemma~\ref{lem:12.5.27.17.21} $\exp_{|N^+S} \colon
N^+S\to \Omega$ is a $C^\infty$ bijection, and $r$ is well-defined as
the distance function on $\Omega$:
$r(x)=\mathop{\mathrm{dist}}(x,\Theta)$, $x\in\Omega$.  The following
lemma implies that the pair $(r,\sigma)$ in fact defines local
coordinates for $\Omega$.

\begin{lemma}
  The exponential map $\exp_{|N^+S} \colon N^+S\to \Omega$ is a
  diffeomorphism.
\end{lemma}
\begin{proof}
  Parts of the arguments below  depend on \cite[Section III.6]{C}.  By
  Lemma~\ref{lem:12.5.27.17.21} it suffices to show that
  $\exp_{|N^+S}$ is a local diffeomorphism.  For $r\ge 0$ and
  $\sigma\in S$ let $\gamma(\cdot;r,\sigma)$ be the geodesic defined
  by
  \begin{align*}
    \gamma(t;r,\sigma)=\exp (tr\nu(\sigma)),\quad t\in [0,1],
  \end{align*}
  and we consider the vector field $Y_\alpha$ along it:
  \begin{align*}
    Y_\alpha(t)=\partial_\alpha\gamma(t;r,\sigma);\quad
    \partial_\alpha=\partial_{\sigma^\alpha},\ \alpha=2,\dots,d.
  \end{align*}
  The vector field $Y_\alpha$ is so-called the \textit{Jacobi field}
  and satisfies the equation
  \begin{align}
    \nabla_{\gamma'}^2Y_\alpha+R(\gamma',Y_\alpha)\gamma'=0\label{eq:12.5.28.6.51}
  \end{align}
  with the initial conditions
  \begin{align}
    Y_\alpha(0)=\partial_\alpha,\quad
    (\nabla_{\gamma'}Y_\alpha)(0)=\pi_{(TS)_\sigma}(r\nabla_\alpha\nu).\label{eq:12.5.28.6.52}
  \end{align}
  Let $X_\alpha(t)$ be the parallel transport of $\partial_\alpha\in
  (TS)_\sigma$ along $\gamma$, i.e.\
  \begin{align*}
    (\nabla_{\gamma'}X_\alpha)(t)=0,\quad X_\alpha(0)=\partial_\alpha,
  \end{align*}
  and seek for a solution to (\ref{eq:12.5.28.6.51}) and
  (\ref{eq:12.5.28.6.52}) of the form
  $Y_\alpha(t)=c_\alpha^\beta(t)X_\beta(t)$.  Since $R=0$ and
  $(\nabla_\alpha \nu)_\beta=((\nabla^2r)_{|S})_{\alpha\beta}$, we
  have (\ref{eq:12.5.28.6.51}) and (\ref{eq:12.5.28.6.52}) reduced to
  \begin{align*}
    (c_\alpha^\beta)''(t)=0,\quad
    c_\alpha^\beta(0)=\delta_\alpha^\beta,\quad
    (c_\alpha^\beta)'(0)=r((\nabla^2r)_{|S})_{\alpha\gamma}(g_{|S})^{\gamma\beta}.
  \end{align*}
  We can solve this as a matrix equation, and hence obtain
  \begin{align}
    Y_\alpha(t)=X_\alpha(t)+tr
      ((\nabla^2r)_{|S})_{\alpha\gamma}(g_{|S})^{\gamma\beta}X_\beta(t).
    \label{eq:12.5.29.15.33}
  \end{align}
  Note that we can choose local coordinates $\sigma$ such that
  $\partial_\alpha$, $\alpha=2,\dots,d$, are principal directions of
  $S$, so that $(\nabla^2r)_{|S}$ and $g_{|S}$ are written as diagonal
  matrices.  Then by positivity (\ref{eq:12.5.27.15.29}) it is
  straightforward to see that $Y_\alpha(t)$, $\alpha=2,\dots ,d$, do
  not degenerate for all $t\ge 0$.  Thus $\exp_{|N^+S}$ is a local
  diffeomorphism.
\end{proof}

\begin{proof}[Proof of Proposition
  \ref{thm:12.5.27.16.50}]
By (\ref{eq:12.5.29.15.33}) we can write the metric of $\Omega$ in
terms of coordinates $(r,\sigma)$, and
\begin{align}
  g=\mathrm{d}r\otimes \mathrm{d}r+
  (g_{|S}+r(\nabla^2r)_{|S})_{\alpha\gamma}(g_{|S})^{\gamma\delta}
  (g_{|S}+r(\nabla^2r)_{|S})_{\delta\beta}\,\mathrm{d}\sigma^\alpha\otimes\mathrm{d}\sigma^\beta.
  \label{eq:12.5.28.8.26}
\end{align}
This can be verified by choosing local coordinates $\sigma^\alpha$
diagonalizing $(\nabla^2r)_{|S}$ and $g_{|S}$ and using the fact that
the parallel transport does not change the length of vectors.  By the
representation (\ref{eq:12.5.28.8.26}) we can compute explicitly
\begin{align*}
  (\nabla^2r)_{\alpha\beta}=((\nabla^2r)_{|S})_{\alpha\beta}
  +r((\nabla^2r)_{|S})_{\alpha\gamma}(g_{|S})^{\gamma\delta}((\nabla^2r)_{|S})_{\delta\beta}.
\end{align*}
Then for any $c<1$ there exists $r_0>0$ such that for all $r\ge r_0$
\begin{align*}
  (\nabla^2r)_{|S_r}\ge cr^{-1}g_{|S_r}.
\end{align*}
Hence we have (\ref{eq:12.5.27.14.14}).  By (\ref{eq:12.5.28.8.26}) we
can compute $|\nabla^kr|$ in $(r,\sigma)$ coordinates, and then the
verification of (\ref{eq:12.5.27.14.15}) is also straightforward.
\end{proof}

\end{document}